\documentclass[journal]{IEEEtran}
\usepackage{amsthm}
\usepackage{amssymb}
\usepackage{graphicx}
\usepackage{amsmath}
\usepackage{newtxmath}
\usepackage[cal=cm]{mathalpha} 
\usepackage{tablefootnote} 
\usepackage{threeparttable} 
\usepackage{bm}
\newtheorem{theorem}{Theorem}{}
{}
{}
{}
{}
{}

\usepackage{algorithm}
\usepackage{algpseudocode}
\usepackage{graphics}
\usepackage{epsfig}
\usepackage{multirow}
\usepackage{caption}
\usepackage{array}
\usepackage{cite}
\usepackage{stfloats}
\usepackage{subfig} 
\usepackage{booktabs}
\usepackage{comment}
\usepackage{color}
\usepackage{tabu} 
\usepackage{stfloats}
\usepackage{mathtools}
\usepackage{stmaryrd} 
\usepackage[
colorlinks=true,
linkcolor=blue,
anchorcolor=blue,
citecolor=blue,
bookmarks=true
]{hyperref}
\usepackage{etoolbox}
\makeatletter

\let\OriginalCite\cite

\newcommand{\MarkFirstCitation}[1]{%
	\ifcsundef{FirstCite@#1}{%
		\csgdef{FirstCite@#1}{1}%
		\Hy@raisedlink{%
			\hyper@anchorstart{FirstCite:#1}%
			\hyper@anchorend
		}%
	}{}%
}
\renewcommand{\cite}[2][]{%
	\ifstrempty{#1}{%
		\OriginalCite{#2}%
	}{%
		\OriginalCite[#1]{#2}%
	}%
	\forcsvlist{\MarkFirstCitation}{#2}%
}
\let\OldBibLabel\@biblabel
\def\CurrentBibKey{}

\renewcommand{\@biblabel}[1]{%
	\ifx\CurrentBibKey\@empty
	\OldBibLabel{#1}%
	\else
	\ifcsname FirstCite@\CurrentBibKey\endcsname
	\begingroup
	\hypersetup{linkcolor=black}%
	\hyperlink{FirstCite:\CurrentBibKey}{%
		\OldBibLabel{#1}%
	}%
	\endgroup
	\else
	\OldBibLabel{#1}%
	\fi
	\fi
}
\let\OldAtBibitem\@bibitem
\def\@bibitem#1{%
	\gdef\CurrentBibKey{#1}%
	\OldAtBibitem{#1}%
}

\let\OldAtLbibitem\@lbibitem
\def\@lbibitem[#1]#2{%
	\gdef\CurrentBibKey{#2}%
	\OldAtLbibitem[#1]{#2}%
}

\makeatother
\usepackage[capitalize]{cleveref} 
\crefname{figure}{Fig.}{Figs.}   
\Crefname{figure}{Fig.}{Figs.}   
\usepackage{subcaption}
\makeatletter
\let\sum\relax
\let\prod\relax
\DeclareSymbolFont{largesymbols}{OMX}{cmex}{m}{n}
\DeclareMathSymbol{\sum}{\mathop}{largesymbols}{"50}
\DeclareMathSymbol{\prod}{\mathop}{largesymbols}{"51}
\DeclareMathOperator*{\diag}{\text{diag}}

\DeclareMathOperator*{\tr}{\text{tr}}

\DeclareMathOperator*{\var}{\mathrm{var}}
\DeclareMathOperator*{\cov}{\mathrm{cov}}

\begin{document}
	\title{Movable Antenna-Enabled Multi-Target Sensing: Sidelobe Suppression via Position Optimization}
	
	
	\author{\IEEEauthorblockN{
			Chengzhi Ye, \IEEEmembership{Student Member,~IEEE}, 
			Ruoyu Zhang, \IEEEmembership{Senior Member,~IEEE},
			Lipeng Zhu, \IEEEmembership{Senior Member,~IEEE},
			Qingqing Wu, \IEEEmembership{Senior Member,~IEEE},
		}
		\thanks{
			
		Chengzhi Ye and Ruoyu Zhang are with the Key Laboratory of Near-Range RF Sensing ICs and Microsystems (NJUST), Ministry of Education, School of Electronic and Optical Engineering, Nanjing University of Science and Technology, Nanjing 210094, China. (e-mail: qxsycz8166@njust.edu.cn; ryzhang19@njust.edu.cn). 
		Lipeng Zhu is with the State Key Laboratory of CNS/ATM and the School
		of Interdisciplinary Science, Beijing Institute of Technology, Beijing 100081,
		China (e-mail: zhulp@bit.edu.cn). Qingqing Wu is with the Department of Electronic Engineering, Shanghai Jiao Tong University, Shanghai 200240, China (email: qingqingwu@sjtu.edu.cn). 
			\textit{(Corresponding author: Ruoyu Zhang)}}
		\vspace{-2em}
	}

	
	\maketitle
	
	

	%
	\IEEEpeerreviewmaketitle

		\begin{abstract}
			Movable antenna (MA) has emerged as a promising technology to unlock spatial degrees of freedom for wireless sensing systems. However, optimizing MA positions for multi-target sensing remains challenging due to the complicated coupling between the array geometry and unknown target directions, as well as severe sidelobe interference. In this paper, we propose an MA-enabled multi-target wireless sensing scheme to achieve both high-precision and robust target parameter estimation. To characterize the sensing performance while maintaining tractability, we introduce two complementary metrics, namely a lower bound (LB) of the Cram\'er-Rao bound (CRB) and the steering vector correlation (SVC). Specifically, by investigating the partial order relationship of the Fisher information matrix, we derive a closed-form LB of the multi-target CRB, which decouples the unknown target directions. Meanwhile, the SVC is employed to characterize the spatial correlation and sidelobe-induced interference among sensing signals over different directions. Based on these two metrics, we formulate an MA position optimization problem that minimizes the SVC while constraining the LB below a predefined threshold. To render the resulting non-convex problem tractable, we exploit the spatial symmetry of the continuous SVC function and discretize the sidelobe region, thereby transforming the spatial minimax objective into a finite set of constraints. Meanwhile, the LB constraint is equivalently reformulated as a spatial variance constraint on the MA positions. Building on these transformations, we develop a successive convex approximation (SCA) algorithm within an alternating optimization (AO) framework to iteratively optimize the horizontal and vertical MA coordinates. 
			Numerical results demonstrate that the proposed MA design significantly suppresses the SVC sidelobes while achieving superior estimation accuracy and enhanced sensing robustness compared with conventional configurations.
		\end{abstract}
	\section{Introduction}
	
	The deployment of future sixth-generation (6G) wireless networks facilitates innovative applications requiring exceptional environmental perception and highly accurate target localization \cite{Saad2020AVision}. To efficiently address these demands, integrated communications and sensing (ISAC) has emerged as a revolutionary architecture utilizing shared resources to make wireless sensing a fundamental capability for identifying surrounding objects \cite{RuoyuDOA}.
	To elevate sensing capabilities, multiple-input multiple-output (MIMO) technologies have been extensively investigated \cite{faisal2020ultramassive,Zhang2022Tensor}. Increasing the number of antennas provides more spatial degrees of freedom (DoF), allowing the system to resolve and estimate more targets, while a larger physical aperture improves spatial resolution and sensing accuracy \cite{DaiQLMIMOOFDM}. Over the past decades, antenna array technologies have evolved from conventional MIMO to massive MIMO, driven by the increasing demands for higher spectral efficiency and more accurate spatial sensing \cite{Zhang2024TWC}. This development has further extended to high-frequency ISAC systems, where large-scale antenna arrays provide enhanced spatial processing capabilities \cite{RZhangTerahertzIntegrated2025}. More recently, extremely large-scale MIMO has emerged to further expand the physical array aperture, offering new opportunities for high-resolution sensing and fine-grained spatial information acquisition \cite{RuoyuDynamic}. Nevertheless, continuously scaling up the number of antennas inevitably incurs substantial hardware costs and associated radio frequency (RF) chains, energy consumption, and computational complexity \cite{GaoLowComplexityMIMO, HaiquanXLMIMO}.
	
	
	Instead of increasing the antenna number, sparse arrays provide an alternative way to enlarge the physical aperture with a limited number of antenna elements \cite{ye2026rotatableantennaenhancedwirelesssensing,ZhuMAsurvey}. By allowing selected inter-element spacings to exceed half a wavelength, sparse arrays can achieve enhanced angular resolution and increased DoF with fewer antennas. Representative designs include nested arrays, difference arrays, and coprime arrays \cite{LiQiangAccurate2019,PiyaVaidyanathan2010,NestedMIMO}. Specifically, the nested arrays combine multiple uniform linear subarrays to increase the available DoF for direction-of-arrival (DOA) estimation \cite{PiyaVaidyanathan2010}, while the hole-free difference coarray can further improve the effective virtual aperture \cite{NestedMIMO}. Coprime arrays exploit coprime spatial samplers to realize high-resolution sensing with a reduced antenna number \cite{VaidyanathanCoPrime}. Despite these advantages, conventional sparse arrays rely on predetermined and highly structured antenna geometries. Once deployed, their physical locations remain fixed and cannot be adapted to changing sensing environments or different performance requirements. From the perspective of array geometry design, this limitation motivates a reconfigurable form of sparse/nonuniform array in which the antenna locations can be dynamically optimized rather than predetermined.
	
	Movable antenna (MA) provides such spatial reconfigurability by enabling antenna elements to flexibly adjust their positions within spatial regions \cite{MaMAsurvey,C11456856}. Different from conventional sparse arrays with fixed geometries, MAs introduce continuous position variables as additional spatial DoF, allowing the array geometry and the corresponding field response to be adapted to the propagation environment and system objective \cite{ModelingMA,LWangMovableAntenna2025,10944486}. In wireless communications, such flexibility provides three fundamental performance gains, namely spatial diversity, beamforming, and spatial multiplexing gains. Specifically, MA position adjustment enables the selection of favorable channel conditions to enhance spatial diversity, improves beamforming performance by reshaping the spatial channel responses, and facilitates spatial multiplexing by reducing channel correlation among different data streams \cite{MaMAsurvey,C11456856,6DMovableAntennaShaoXiaodan2025}. These advantages have motivated extensive studies on MA-enabled wireless communications, including channel acquisition for MA-MIMO systems \cite{ZhangChannel2024} and joint antenna position and power optimization for multiuser uplink transmission \cite{Zhu2023Movable}.
	
	The spatial flexibility of MAs has also been extensively exploited to improve wireless sensing performance \cite{WangNearfieldMA,6DSensing,LiJinTRMAsensing,YuchenMASensing,ChengzhiTWCMIMOOFDM}. In particular, MAs offer several fundamental advantages for wireless sensing. First, they can enlarge the effective array aperture to enhance angular resolution without increasing the number of antennas. Second, the flexible reconfiguration of array geometries enables the spatial response to be optimized for improved estimation accuracy and reduced angular ambiguity \cite{maMovableAntennaEnhanced2024}. Third, antenna movement during sensing, when permitted, provides additional spatial-temporal DoF for enhancing sensing performance \cite{ma2025movableantennatrajectory}. These advantages have motivated extensive investigations into MA-enabled sensing. Specifically, the authors in \cite{maMovableAntennaEnhanced2024} optimized MA positions to minimize the Cram\'er-Rao bound (CRB) for DOA estimation, while the authors in \cite{ChengzhiIOTJ20262D} jointly optimized the shape of the movable region and the antenna positions to further improve estimation accuracy. Furthermore, \cite{ma2025movableantennatrajectory} investigated the fundamental CRB scaling laws over time and space through continuous MA trajectory optimization. Besides position optimization, practical sensing robustness under array imperfections has also been investigated \cite{DOAGPEChengzhi2025,ye2026selfcalibrationdoaestimationmovable}, including gain and phase error compensation and antenna position self-calibration. Finally, MAs can flexibly adapt their array configurations to different communication and sensing requirements. In this regard, the authors in \cite{MaISACposition} optimized MA positions in an ISAC system to enlarge the achievable trade-off region between communication rate and sensing performance.

	Despite these noteworthy advances, extending sensing-oriented MA position optimization to unknown multi-target scenarios remains challenging. Recently, the authors in \cite{mao2025movableantennaarrayenhancedmultitarget} characterized the multi-target CRB as a function of MA positions and optimized the expected CRB over a prescribed statistical distribution of target directions to improve estimation accuracy. Furthermore, the authors in \cite{chen2026exploitingmovableantennasmultitarget} investigated multi-target sensing in multipath environments, where the posterior probabilities of target locations were exploited to construct a weighted angular ambiguity function for MA position optimization, aiming to suppress sidelobes and narrow the main lobe. Although these studies have demonstrated the potential of MAs for multi-target sensing, the former relies on the statistical distribution of target directions, while the latter requires posterior information obtained from preliminary target estimation. Therefore, designing MA arrays to suppress spatial ambiguity without relying on such target-specific information remains an important problem.
	
	More importantly, high-resolution multi-target sensing requires not only a sufficiently large effective array aperture but also effective suppression of spatial ambiguity across the entire sensing region. For spectrum-based DOA estimation, different target directions correspond to different steering vectors, whose correlations characterize the ambiguity and sidelobe properties of the array response. Excessively high correlations may introduce strong sidelobes and ambiguous spatial responses, potentially masking weak targets or producing false peaks, particularly under low signal-to-noise ratio (SNR) conditions. Such effects can become more pronounced in sparse or highly nonuniform array configurations with large spatial gaps. Motivated by these observations, we consider a target-direction-independent MA position optimization framework that minimizes the maximum steering vector correlation (SVC) over a prescribed sidelobe region, without requiring prior knowledge of the actual target DOAs or their location posterior probabilities. Meanwhile, an aperture-related constraint derived from a lower bound (LB) on the multi-target CRB is imposed to maintain a minimum spatial spread of the MA array.
	
	Motivated by the above observations, in this paper, we investigate a multi-target wireless sensing scheme based on an MA array to jointly achieve high-precision parameter estimation and robust spatial interference suppression. The main contributions are summarized as follows:

	\begin{enumerate}
		\item We establish the system model for wireless sensing and theoretically derive a closed-form lower bound for the highly coupled multi-target CRB. By exploiting the partial order relationship of projection matrices within the Fisher information matrix (FIM), this derivation successfully decouples the unknown target angles to establish a tractable theoretical limit for the estimation variance. Furthermore, we introduce the SVC as a complementary metric to accurately quantify the spatial interference suppression capability.
		\item Guided by the proposed metrics, we formulate an optimization problem to design the MA positions to minimize the SVC subject to a specific CRB-LB constraint. To render this intractable spatial minimax optimization solvable, we utilize the symmetric property of the continuous SVC function to transform it into a finite set of manageable discrete constraints through spatial sampling and auxiliary variables. Simultaneously, the complex CRB-LB performance requirement is equivalently reformulated into a straightforward spatial variance constraint.
		\item To solve the optimization problem, we develop a successive convex approximation (SCA)-based AO algorithm, which optimizes the $x$-axis and $y$-axis of the MA positions iteratively. Specifically, within each decomposed subproblem, we derive a Lipschitz constant independent of the spatial locations to construct a convex surrogate upper bound for the SVC function, thereby guaranteeing the convergence of the algorithm.
		\item We conduct numerical simulations to evaluate the proposed solution. For fixed target directions, we compare the multiple signal classification (MUSIC) spatial spectrum and the root mean square error (RMSE) across different array configurations. Moreover, we validate the statistical robustness of the proposed MA via estimation success rates under randomized target directions. The results verify that the proposed MA achieves high estimation accuracy and sensing robustness in multi-target scenarios.
	\end{enumerate}
	
	\section{System Model}

	\begin{figure}[t]
		\centering
		\includegraphics[width=0.85\linewidth]{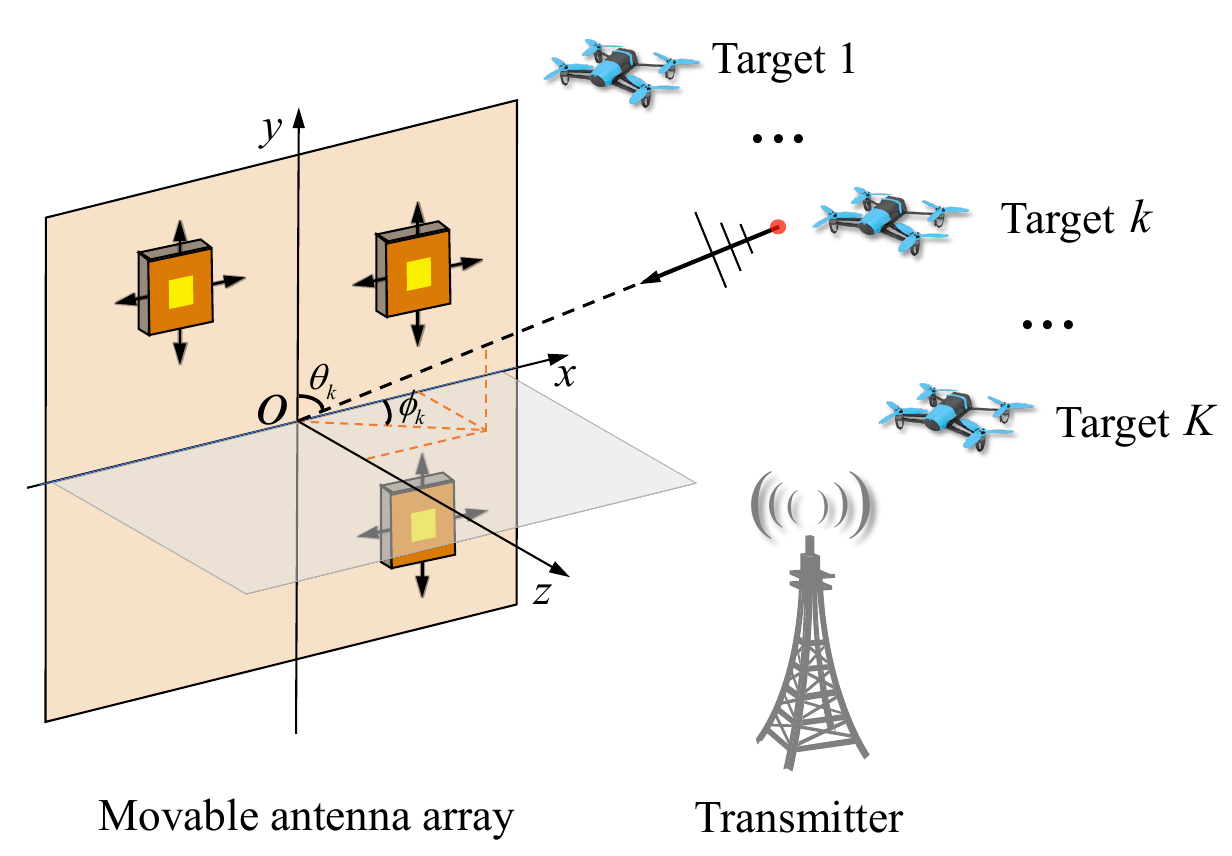}
		\caption{Multi-target sensing with MAs.}
		\label{FigRA}
		\vspace{-0.5em}
	\end{figure}
	\subsection{System Model}
	As shown in Fig. \ref{FigRA}, we consider a
	bistatic wireless sensing system
	with $K$ targets, consisting of a
		transmitter and a receiving base station.
	The transmitter employs a single
		omnidirectional antenna to illuminate the targets,
		while the receiving base station collects the
		target-scattered signals using $N$ MAs.
		The received signals from different targets are
		assumed to be mutually uncorrelated.
	The receiving base station is equipped
	with $N$ MAs that can move freely within the region
	$\mathcal{C}=[0, A_x]\times [0, A_y]$.
	The MA positions are denoted by
	$\bm{p} = [\bm{x}^T, \bm{y}^T]^T
	\in \mathbb{R}^{2N \times 1}$,
	where $\bm{x}=[x_1, \ldots, x_N]^T
	\in \mathbb{R}^{N \times 1}$ and
	$\bm{y}=[y_1, \ldots, y_N]^T
	\in \mathbb{R}^{N \times 1}$
	represent the $x$-axis and $y$-axis MA positions,
	respectively. For convenience, let
	$\bm{p}_n = [x_n,y_n]^T
	\in \mathbb{R}^{2 \times 1}$
	denote the position coordinates of the $n$-th MA.
	The DOA at
		the receiving base station of the $k$-th target
	is expressed as
	$\bm{r}_k=[u_k, v_k]^T
	\in \mathbb{R}^{2 \times 1}$,
	where $u_k = \sin\theta_k\cos\phi_k
	\in [-u_{\max}, u_{\max}]$
	and $v_k=\cos\theta_k
	\in [-v_{\max}, v_{\max}]$,
	with $\theta_k$ and $\phi_k$ denoting the
	elevation and azimuth angles of the $k$-th target
	with respect to the receiving base
		station, respectively. Consequently, the steering vector of the $k$-th target can be formulated as
	\begin{align}
		\bm{a}(\bm{p}, \bm{r}_k) = \Big[\mathrm{e}^{\mathrm{j}\kappa \bm{p}_1^T\bm{r}_k}, \ldots, \mathrm{e}^{\mathrm{j}\kappa \bm{p}_N^T\bm{r}_k}\Big]^T\in \mathbb{C}^{N \times 1},
	\end{align}
	where $\kappa = \frac{2\pi}{\lambda}$ represents the wavenumber. Then the received signal at the $t$-th snapshot, can be expressed as
	\begin{align}
		\bm{y}(t) = \sum_{k=1}^{K}\bm{a}(\bm{p}, \bm{r}_k) s_k(t)+\bm{n}(t), 1\leq t\leq T,
	\end{align}
	where $s_k(t) \in \mathbb{C}$ represents the signal of the $k$-th target at the $t$-th snapshot, and $\bm{n}(t)= [n_1(t), \ldots, n_N(t)]^T\sim \mathcal{C}\mathcal{N}(\bm{0}_N, \sigma_n^2\bm{I}_N)$ denotes the additive white Gaussian noise (AWGN) at the $t$-th snapshot with $\sigma_n^2$ denoting the noise power. For all $T$ snapshots, the received signal matrix can be given by
	\begin{align}
		\bm{Y} = \bm{A}\bm{S}+\bm{N},\label{YY}
	\end{align}
	where $\bm{Y}= [\bm{y}(1), \ldots, \bm{y}(T)]\in \mathbb{C}^{N\times T}$, $\bm{S}=[\bm{s}(1), \ldots, \bm{s}(T)]\in \mathbb{C}^{K \times T}$, $\bm{s}(1)=[s_1(1), \ldots, s_K(1)]\in \mathbb{C}^{K \times 1}$, $\bm{N}=[\bm{n}(1), \ldots, \bm{n}(T)] \in \mathbb{C}^{N \times T}$, and $\bm{A} = [\bm{a}(\bm{p}, \bm{r}_1) ,\ldots, \bm{a}(\bm{p}, \bm{r}_K) ] \in \mathbb{C}^{N \times K}$. From \eqref{YY}, we obtain the covariance matrix $\bm{R}_{\bm{Y}}\in \mathbb{C}^{N \times N}$ of the received signal $\bm{Y}$ as
	\begin{align}
		\bm{R}_{\bm{Y}} = \mathbb{E}\{\bm{Y}\bm{Y}^H\}= \bm{A} \bm{R}_{\bm{S}}\bm{A}^H +\sigma_n^2 \bm{I}_N,
	\end{align}
	where $\bm{R}_{\bm{S}}=\diag(\sigma_1^2, \ldots,\sigma_K^2)\in \mathbb{R}^{K \times K}$ denotes the signal covariance matrix. $\sigma_k^2$ represent the signal power of the $k$-th target, respectively. Applying eigenvalue decomposition to $\bm{R}_{\bm{Y}}$ yields
	\begin{align}
		\bm{R}_{\bm{Y}}=\sum_{n=1}^{N}\xi_n \bm{e}_n \bm{e}_n^H=\bm{E}_{\bm{S}}\bm{\varXi}_{\bm{S}}\bm{E}_{\bm{S}}^H + \bm{E}_{\bm{N}} \bm{\varXi}_{\bm{N}} \bm{E}_{\bm{N}}^H,
	\end{align}
	where $\xi_n$ denotes the $n$-th largest eigenvalue, and $\bm{e}_n \in \mathbb{C}^{N \times 1}$ indicates the corresponding eigenvector. $\bm{E}_{\bm{S}}=[\bm{e}_1, \ldots, \bm{e}_{K}] \in \mathbb{C}^{N \times K}$ and $\bm{E}_{\bm{N}}=[\bm{e}_{K+1}, \ldots, \bm{e}_{N}] \in \mathbb{C}^{N \times (N-K)}$ represent the signal subspace and the noise subspace, respectively. Then the multi-target DOA estimation based on the MUSIC algorithm requires searching for the $K$ largest peaks in the following MUSIC spatial spectrum:
	\begin{align}
		P(\bm{r}) = \frac{1}{\bm{a}^H(\bm{p}, \bm{r})\bm{E}_{\bm{N}}\bm{E}_{\bm{N}}^H \bm{a}(\bm{p}, \bm{r})},
	\end{align}
	where $\bm{r} \in [-u_{\max}, u_{\max}]\times [-v_{\max}, v_{\max}]$ represents the possible target direction. By finding the estimated spatial direction vectors $\{\bm{\hat{r}}_k= [\hat{u}_k, \hat{v}_k]^T\}_{k=1}^{K}$ corresponding to the $K$ peaks of $P(\bm{r})$, the $k$-th DOA can be estimated as
	\begin{align}
		\hat{\theta}_k = \arccos(\hat{v}_k), \quad
		\hat{\phi}_k =
		\arccos\Big(
		\frac{\hat{u}_k}{\sqrt{1-\hat{v}_k^2}}
		\Big),
	\end{align}
	where $\hat{\theta}_k$ and $\hat{\phi}_k$ denote the estimated elevation and azimuth angles of the $k$-th target, respectively. Next, we will introduce the performance metrics of the wireless sensing system to optimize the MA positions.
	
	\vspace{-0.5em}
	\section{Performance Metrics}
	In this section, we present two complementary evaluation metrics for wireless sensing. The first metric is the LB derived from the CRB to establish a theoretical limit for the estimation accuracy. Since the LB inherently omits the mutual interference among multiple targets, we further introduce the SVC as a complementary metric to measure the sidelobe levels and evaluate the interference suppression capability.
	\vspace{-0.5em}
	\subsection{CRB Lower Bound}
	It is well known that the CRB characterizes the lower bound for the performance of unbiased parameter estimation, which in the considered wireless sensing system can be specifically expressed as
	\begin{align}\label{CRBpw}
		\!\!\!\!{\mathrm{CRB}(\bm{p}, \bm{\omega})} =\frac{\sigma_n^2}{2}\tr \Big\{\Big[\mathfrak{R}\Big\{\sum_{t=1}^{T}{\bm{\tilde{S}}^H(t)\bm{W}^H\bm{P}_{\bm{A}}^{\perp}}\bm{W}\bm{\tilde{S}}(t)\Big\}\Big]^{-1}\Big\},
	\end{align}
	where the detailed derivation and definitions of the relevant variables are provided in Appendix \ref{AppendixA}. However, since $\bm{P}_{\bm{A}}^{\perp} = \bm{I}_N - \bm{A}(\bm{A}^H\bm{A})^{-1}\bm{A}^H$ contains the unknown angles, it is impossible to separate a function from the CRB that depends solely on the MA positions and is independent of the unknown angles. This coupling makes direct CRB-based optimization intractable, and we therefore derive an LB on the CRB, which can be shown in Theorem \ref{theorem1}.
	
	\begin{theorem}\label{theorem1}
		For the considered wireless sensing system, an LB on the CRB can be given by
		\begin{align}
			\!\!\!\!\!{\mathrm{CRB}(\bm{p}, \bm{\omega})}\ge \mathrm{LB}(\bm{p}, \bm{\omega}) =\frac{1}{\var(\bm{y}) +\var(\bm{x})} \sum_{k=1}^{K} \frac{\sigma_n^2\lambda^2}{2\pi^2 TN \sigma_k^2}.
		\end{align}
		where $\bm{\omega} = [\bm{u}^T, \bm{v}^T]^T$.
	\end{theorem}
	\begin{proof}
		See Appendix \ref{AppendixB}.
	\end{proof}
	Theorem \ref{theorem1} successfully decouples the unknown target angles from the highly coupled CRB matrix and establishes a tractable closed-form lower bound. It explicitly demonstrates that the theoretical estimation error limit is inversely proportional to the spatial variances of the MA coordinates. This mathematical property reveals a critical physical insight that expanding the spatial distribution of the MAs effectively enlarges the array aperture to enhance the overall sensing performance. It is worth noting that imposing an upper threshold
	on the LB does not directly guarantee the actual
	CRB performance. Instead, this constraint ensures
	a minimum spatial spread of the MA array, thereby
	preserving the array aperture while suppressing the
	sidelobes through SVC minimization.
	
	Minimizing the LB fundamentally lowers the theoretical limit of the estimation variance, thereby providing a more favorable theoretical lower bound for target angle estimation. Therefore, enlarging the spatial variances of the MA coordinates can effectively decrease the derived LB and improve the fundamental estimation-accuracy limit.  However, based on Theorem \ref{theorem1}, the LB is derived through a relaxation that omits the cross-terms of the FIM. Consequently, the LB fails to characterize the mutual influence among different targets. Furthermore, the interference originating from multiple targets as well as from other non-target directions within the search area significantly affects the DOA estimation results. Therefore, in the subsequent subsection, we introduce the SVC metric to comprehensively evaluate the spatial resolution performance and the interference suppression capability in multi-target scenarios.
	
	\subsection{Steering Vector Correlation}
	The SVC characterizes the steering vector correlation between the target direction and potential target directions. 
	For the $k$-th target direction $\bm{r}_k$ and a candidate target direction $\bm{r}$ the SVC can be specifically expressed as
	\begin{align}
		\mathrm{SVC}(\bm{p},\bm{r}_k, \bm{r})&= \frac{1}{N^2}|\bm{a}(\bm{p}, \bm{r}_k)^H \bm{a}(\bm{p}, \bm{r})|^2.
	\end{align}
	On the one hand, for a fixed target direction $\bm{r}_k$, the SVC equals unity at $\bm{r}=\bm{r}_k$, while the spatial correlation responses outside a predefined main lobe protection region are treated as sidelobes for array design. Minimizing the maximum SVC over this sidelobe region suppresses high-correlation peaks that may otherwise cause angular ambiguities and false detections, particularly under low-SNR conditions. On the other hand, it mitigates the mutual interference caused by the target's own sidelobes leaking into the spatial directions of other co-existing targets. Consequently, this helps prevent weak targets from being masked by the strong sidelobes of adjacent targets, thereby improving the robustness of multi-target DOA estimation.

	\section{Problem Formulation and Transformation}
	
	\subsection{Sidelobe Region and Problem Formulation}
	By expanding the spatial inner product in the SVC expression, we obtain
	\begin{align}\label{SVChua}
		\mathrm{SVC}(\bm{p},\bm{r}_k, \bm{r}) &= \frac{1}{N^2} \Big|\sum_{n=1}^{N} \mathrm{e}^{\mathrm{j}\kappa \bm{p}_n^T(\bm{r} - \bm{r}_k)}\Big|^2 = \frac{1}{N^2} \Big|\sum_{n=1}^{N} \mathrm{e}^{\mathrm{j}\kappa \bm{p}_n^T\bm{\tilde{r}}}\Big|^2 \nonumber\\
		&= \frac{1}{N^2} \Big|\sum_{n=1}^{N} \mathrm{e}^{\mathrm{j}\kappa \big(x_n\tilde{u} + y_n\tilde{v}\big)}\Big|^2 \triangleq \mathrm{SVC}(\bm{p},\bm{\tilde{r}}),
	\end{align}
	where the second equality defines $\bm{\tilde{r}} = \bm{r} - \bm{r}_k = [\tilde{u}, \tilde{v}]^T$, $\tilde{u} = u - u_k$, and $\tilde{v} = v - v_k$. Since $u, u_k\in [-u_{\max}, u_{\max}]$, we have $\tilde{u} \in [-2u_{\max}, 2u_{\max}]$. Similarly $\tilde{v}\in [-2v_{\max}, 2v_{\max}]$, thus yielding $\bm{\tilde{r}} \in [-2u_{\max}, 2u_{\max}] \times [-2v_{\max}, 2v_{\max}]$. For notational convenience, we write $\mathrm{SVC}(\bm{p},\bm{\tilde{r}})$ for the expression in \eqref{SVChua}. Observing \eqref{SVChua}, we obtain
	\begin{align}
		\mathrm{SVC}(\bm{p}, -\bm{\tilde{r}}) = \frac{1}{N^2} \Big|\sum_{n=1}^{N} \mathrm{e}^{-\mathrm{j}\kappa \bm{p}_n^T\bm{\tilde{r}}}\Big|^2 = \mathrm{SVC}(\bm{p},\bm{\tilde{r}}),
	\end{align}
	which indicates that the SVC is an even function with respect to $\bm{\tilde{r}}$. Moreover, when $\bm{r}_k = \bm{r}$, $\mathrm{SVC}(\bm{0}_2) = 1$, corresponding to the main lobe of the SVC. Our goal is to suppress the sidelobes, thereby requiring the exclusion of the main lobe region. We define the main lobe protection thresholds for $u$ and $v$ as $\gamma_u$ and $\gamma_v$, respectively. Hence, the region where we suppress the sidelobes can be expressed as
	\begin{align}
		\mathcal{D} =& \big\{ (\tilde{u}, \tilde{v}) \mid \tilde{u} \in [0, 2u_{\max}], \tilde{v} \in [-2v_{\max}, 2v_{\max}] \big\}\nonumber\\
		&\setminus \big( [0, \gamma_u] \times [-\gamma_v, \gamma_v] \big),
	\end{align}
	The specific diagram can be seen in \Cref{Show}.
	The region enclosed by the red dashed box denotes the main lobe, serving as the protected area, whereas the white-shaded area represents the optimization region $\mathcal{D}$. Our optimization objective is to minimize the values across this region as much as possible to suppress the peak sidelobe level, thereby effectively mitigating mutual interference among multiple targets.
	\begin{figure}[t]
		\vspace{-0.5em}
		\centering
		\includegraphics[width=0.75\linewidth]{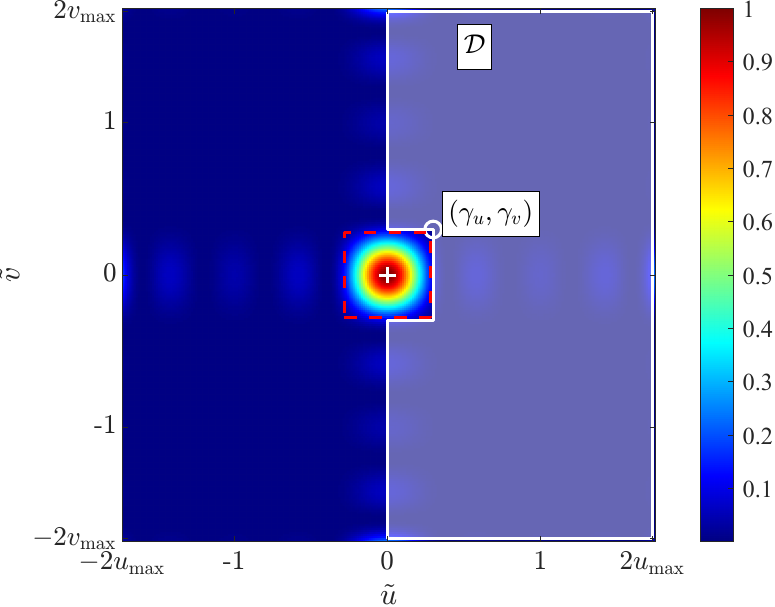}
		\caption{Illustration for the region where we suppress the sidelobes.}
		\label{Show}
		\vspace{-0.5em}
	\end{figure}
	
	Since the estimation performance is affected by both the LB and the SVC, we aim to minimize the maximum SVC over the sidelobe region to suppress the influence of sidelobes on the estimation while guaranteeing that the LB satisfies a specific condition. Therefore, the optimization problem can be formulated as
\begin{subequations}\label{P1}
		\begin{align}
			\mathrm{(P1)}\enspace\min_{\bm{p}}\enspace &\max_{\bm{\tilde{r}}\in\mathcal{D}}\mathrm{SVC}(\bm{p},\bm{\tilde{r}}) \label{SVCrkr}\\
			\mathrm{s.t.} \enspace &\mathrm{LB}(\bm{p}, \bm{\omega})\le \eta \label{CRBC}\\
			& \bm{p}_n \in \mathcal{C}, n = 1, \ldots, N \label{RegionC}\\
			& ||\bm{p}_n - \bm{p}_m || \ge d_{\min}, n\neq m, \label{JianC}
		\end{align}
	\end{subequations}
	where \eqref{CRBC} guarantees that the LB is below a given threshold $\eta$, \eqref{RegionC} specifies the constraint on the movable region, and \eqref{JianC} ensures that the inter-antenna spacing is no less than a minimum distance $d_{\min} = \frac{\lambda}{2}$ to avoid significant mutual coupling. The continuous minimax objective and the non-convex constraints \eqref{CRBC} and \eqref{JianC} make (P1) difficult to solve directly. In the following subsection, we transform the optimization problem.
	
	\vspace{-0.5em}
	\subsection{Problem Transformation}
	To facilitate the construction of the optimization problem, we sample $\mathcal{D}$ discretely to obtain $\mathcal{R} = \{\bm{\tilde{r}}_{l} \}_{l=1}^{L} =\{ [\tilde{u}_l, \tilde{v}_l]^T\}_{l=1}^{L}$, where $\bm{\tilde{r}}_l\in\mathcal{D}$ and $L$ denotes the total number of sampling points. Consequently, the optimization problem can be transformed into
	\begin{subequations}\label{P2}
		\begin{align}
			\mathrm{(P2)}\enspace\min_{\bm{p}, \tau}\enspace &\tau \\
			\mathrm{s.t.} \enspace &\mathrm{SVC}(\bm{p}, \bm{\tilde{r}}_l) \le \tau,  l= 1, \ldots, L \label{SVCP2}\\
			&\eqref{CRBC}, \eqref{RegionC}, \eqref{JianC},
		\end{align}
	\end{subequations}
	where $\tau$ is an auxiliary variable we introduce as an upper bound for all discrete SVC values.
	
	Since the LB is proportional to the reciprocal of $\var(\bm{y}) +\var(\bm{x})$ based on Theorem \ref{theorem1}, we have
	\begin{align}
		\mathrm{LB}(\bm{p}, \bm{\omega})\le \eta\enspace \Leftrightarrow \enspace\var(\bm{y}) +\var(\bm{x})\ge \frac{\sum_{k=1}^{K} {4\varpi_k}}{\eta} \triangleq\delta,
	\end{align}
	where $\delta$ is the corresponding threshold. Based on this, the optimization problem can be further transformed into
	\begin{subequations}\label{P3}
		\begin{align}
			\mathrm{(P3)}\enspace\min_{\bm{p}, \tau}\enspace &\tau \\
			\mathrm{s.t.} \enspace &\mathrm{SVC}(\bm{p}, \bm{\tilde{r}}_l) \le \tau, l= 1, \ldots, L \label{SVCP3}\\
			&\var(\bm{y}) +\var(\bm{x})\ge \delta \label{VarP3}\\
			&\eqref{RegionC}, \eqref{JianC}.
		\end{align}
	\end{subequations}
	Although the problem is simplified into (P3), it remains highly non-convex and intractable to solve directly due to the inherent non-convexity of the SVC function along with the complicated geometric constraints. To tackle this challenging optimization problem, we develop an SCA-based AO algorithm in the subsequent section.
	\vspace{-0.5em}
	\section{Proposed SCA-based AO Algorithm}
	
	In this section, we propose an SCA-based AO algorithm. We iteratively optimize the $\bm{y}$ coordinates of the antennas by utilizing the SCA technique with fixed $\bm{x}$ coordinates, and similarly optimize the $\bm{x}$ coordinates with fixed $\bm{y}$ coordinates. Because the original problem contains highly non-convex fitness objectives and geometric constraints, we construct globally tight convex surrogate functions for the non-convex terms at each iteration using Taylor expansion and second-order scaling. It is worth noting that the corresponding Lipschitz constant is rigorously derived to bound the absolute maximum value of the Hessian matrix. Through this AO approach, we decouple the originally intractable non-convex position optimization problem and transform it into a series of readily solvable convex subproblems for iterative execution until convergence.
	
	By explicitly expressing the MA position constraints,
	problem (P3) can be rewritten as
	\begin{subequations}\label{P4}
		\begin{align}
			\!\!\!\!\mathrm{(P4)}\enspace
			\min_{\bm{x},\bm{y},\tau}\enspace &\tau\\
			\mathrm{s.t.}\enspace
			&\eqref{VarP3},\eqref{SVCP3}\\
			&0\le x_n\le A_x,\enspace n=1,\ldots,N \label{Ax}\\
			&0\le y_n\le A_y,\enspace n=1,\ldots,N \label{Ay}\\
			&\sqrt{(x_n-x_m)^2+(y_n-y_m)^2}
			\ge d_{\min},\enspace n\neq m. \label{JianP4}
		\end{align}
	\end{subequations}
	
	\subsection{Optimizing \texorpdfstring{$\bm{x}$}{w} with Fixed \texorpdfstring{$\bm{y}$}{w}}
	In this subsection, we optimize $\bm{x}$ with fixed $\bm{y}$. Thus, the problem (P4) can be simplified as
	\begin{subequations}\label{P5}
		\begin{align}
			\mathrm{(P5)}\enspace\min_{\bm{x}, \tau}\enspace &\tau \\
			\mathrm{s.t.} \enspace &\mathrm{SVC}(\bm{x}, \bm{\tilde{r}}_l) \le \tau, l= 1, \ldots, L \label{SVCP5}\\
			&\var(\bm{x})\ge \delta_x \label{varP5}\\
			& \eqref{Ax}, \eqref{JianP4},
		\end{align}
	\end{subequations}
	where $\delta_x $ represents the threshold in the $x$ axis. Notice that in problem (P4), \eqref{SVCP5}, \eqref{varP5}, and \eqref{JianP4} are non-convex. We adopt the SCA algorithm to solve them. First, for \eqref{SVCP5}, when fixing $\bm{y}$ we have
	\begin{align}
		\mathrm{SVC}(\bm{x},\bm{\tilde{r}}_l) &= \frac{1}{N^2} \Big|\sum_{n=1}^{N} \mathrm{e}^{\mathrm{j}\kappa (x_n\tilde{u}_l + y_n\tilde{v}_l)}\Big|^2\\
		& = \frac{1}{N^2} \sum_{n=1}^{N} \sum_{m=1}^{N} \mathrm{e}^{\mathrm{j}\kappa \big[  (x_n\tilde{u}_l + y_n\tilde{v}_l) - (x_m\tilde{u}_l + y_m\tilde{v}_l) \big]}\nonumber\\
		& = \frac{1}{N^2} \sum_{n=1}^{N} \sum_{m=1}^{N} \cos (\varphi_{n,m,l}),\nonumber
	\end{align}
	where $\varphi_{n,m,l} = \kappa\big[ (x_n \tilde{u}_l + y_n \tilde{v}_l) - (x_m \tilde{u}_l + y_m \tilde{v}_l) \big]$. At the $i$-th iteration, a quadratic surrogate upper bound is constructed as
	\begin{align}
		\mathrm{SVC}(\bm{x}, \bm{\tilde{r}}_l)&\le \mathrm{SVC}(\bm{x}^{(i)}, \bm{\tilde{r}}_l)+\nabla_{\bm{x}} \mathrm{SVC}(\bm{x}^{(i)}, \bm{\tilde{r}}_l)^T(\bm{x}- \bm{x}^{(i)})\nonumber\\
		&+\frac{Z_{x,l}}{2}||\bm{x}-\bm{x}^{(i)}||_2^2 \triangleq \overline{\mathrm{SVC}}(\bm{x};\bm{x}^{(i)}, \bm{\tilde{r}}_l),
	\end{align}
	where $Z_{x,l}$ represents the Lipschitz constant. The $n$-th element of the gradient vector $\nabla_{\bm{x}} \mathrm{SVC}(\bm{x}^{(i)}, \bm{\tilde{r}}_l)$ can be expressed as
	\begin{align}
		\big[	\nabla_{\bm{x}} \mathrm{SVC}(\bm{x}^{(i)}, \bm{\tilde{r}}_l)\big]_{n} &=  \frac{\partial  \mathrm{SVC}(\bm{x}^{(i)},\bm{\tilde{r}}_l)}{ \partial x_n} \nonumber\\
		&=
		-\frac{2\kappa \tilde{u}_l}{N^2} \sum_{m=1}^{N} \sin (\varphi_{n,m,l}^{(i)}).
	\end{align}
	Furthermore, to guarantee that $\overline{\mathrm{SVC}}(\bm{x};\bm{x}^{(i)}, \bm{\tilde{r}}_l)$ is a global upper bound of $\mathrm{SVC}(\bm{x}, \bm{\tilde{r}}_l)$, $Z_{x,l}$ must satisfy
	\begin{align}
		Z_{x,l} \bm{I}_N\succeq \nabla_{\bm{x}}^2 \mathrm{SVC}(\bm{x}^{(i)}, \bm{\tilde{r}}_l),
	\end{align}
	where $ \nabla_{\bm{x}}^2 \mathrm{SVC}(\bm{x}^{(i)}, \bm{\tilde{r}}_l)$ is the Hessian matrix with respect to $\bm{x}^{(i)}$. Its element can be given as \eqref{nn} and \eqref{nm}.
	\begin{figure*}
		\vspace{-0.5em}
		\begin{align}\label{nn}
		\big[\nabla_{\bm{x}}^2 \mathrm{SVC}(\bm{x}^{(i)}, \bm{\tilde{r}}_l)\big]_{n, n} &= \frac{\partial^2 \mathrm{SVC}}{\partial x_n^2} = \frac{\partial}{\partial x_n} \Big[ -\frac{2\kappa \tilde{u}_l}{N^2} \sum_{m \neq n}^{N} \sin(\varphi_{n,m,l}^{(i)}) \Big] = -\frac{2\kappa^2 \tilde{u}_l^2}{N^2} \sum_{m \neq n}^{N} \cos(\varphi_{n,m,l}^{(i)}),
\\\label{nm}
		\big[\nabla_{\bm{x}}^2 \mathrm{SVC}(\bm{x}^{(i)}, \bm{\tilde{r}}_l)\big]_{n, m} &= \frac{\partial^2 \mathrm{SVC}}{\partial x_n \partial x_m} = \frac{\partial}{\partial x_m} \Big[ -\frac{2\kappa \tilde{u}_l}{N^2} \sin(\varphi_{n,m,l}^{(i)}) \Big] = \frac{2\kappa^2 \tilde{u}_l^2}{N^2} \cos(\varphi_{n,m,l}^{(i)}) , \enspace n \neq m.
	\end{align}
	\vspace{-0.5em}
		\hrule
		\vspace{-0.5em}
	\end{figure*}
	Defining the symmetric cosine weight matrix $\bm{B}_l^{(i)} \in \mathbb{R}^{N \times N}$ with elements $[\bm{B}_l^{(i)}]_{n,m} = \cos(\varphi_{n,m,l}^{(i)})$, this Hessian matrix $\nabla_{\bm{x}}^2 \mathrm{SVC}(\bm{x}^{(i)}, \bm{\tilde{r}}_l) \in \mathbb{R}^{N \times N}$ can be written in the following matrix form
	\begin{align}
		\nabla_{\bm{x}}^2 \mathrm{SVC}(\bm{x}^{(i)}, \bm{\tilde{r}}_l) = \frac{2\kappa^2 \tilde{u}_l^2}{N^2} \Big[ \bm{B}_l^{(i)} - \diag\big(\bm{B}_l^{(i)}\bm{1}_N\big) \Big].
	\end{align}
	To construct a tractable convex surrogate function and guarantee the monotonic convergence of the SCA algorithm, it is essential to determine the upper bound of the objective's curvature. Next, we obtain a position-independent Lipschitz constant via Theorem \ref{theorem2}.
	\begin{theorem}\label{theorem2}
		The Hessian matrix of $\mathrm{SVC}(\bm{x}^{(i)}, \bm{\tilde{r}}_l)$ with respect to $\bm{x}^{(i)}$ is globally upper-bounded by a diagonal matrix
		\begin{align}
			\nabla_{\bm{x}}^2 \mathrm{SVC}(\bm{x}^{(i)}, \bm{\tilde{r}}_l) \preceq \frac{2\kappa^2 \tilde{u}_l^2}{N} \bm{I}_N.
		\end{align}
		Based on this, the Lipschitz constant with respect to $\nabla_{\bm{x}}^2 \mathrm{SVC}(\bm{x}^{(i)}, \bm{\tilde{r}}_l)$ is $Z_{x,l} = \frac{2\kappa^2 \tilde{u}_l^2}{N}$.
	\end{theorem}
	\begin{proof}
		Let the core structure of the Hessian matrix be denoted as $\bm{M}^{(i)} = \bm{B}_l^{(i)} - \mathrm{diag}\big(\bm{B}_l^{(i)}\bm{1}_N\big)$. For any real non-zero vector $\bm{c} = [q_1, q_2, \ldots, q_N]^T \in \mathbb{R}^{N \times 1}$, the quadratic form of this matrix can be simplified and bounded as follows:
		\begin{align} \bm{c}^T \bm{M}^{(i)} \bm{c} &= -\frac{1}{2} \sum_{n=1}^{N} \sum_{m=1}^{N} (c_n - c_m)^2 \cos(\varphi_{n,m,l}^{(i)})\nonumber \\ &\le \frac{1}{2} \sum_{n=1}^{N} \sum_{m=1}^{N} (c_n - c_m)^2 = N \sum_{n=1}^{N} c_n^2 - \Big( \sum_{n=1}^{N} c_n \Big)^2\nonumber \\ &\le N \|\bm{c}\|_2^2, 
		\end{align}
		where the first inequality holds because $\cos(\cdot) \ge -1$ for any phase difference and the perfect square term satisfies $(q_n - q_m)^2 \ge 0$. The second inequality holds because the squared sum term $\big( \sum_{n=1}^{N} q_n \big)^2 \ge 0$.
		
		According to the Rayleigh Quotient theorem, the bounded quadratic form $\bm{c}^T \bm{M}^{(i)} \bm{c} \le N \|\bm{c}\|_2^2$ strictly proves that the structure matrix satisfies $\bm{M}^{(i)} \preceq N \bm{I}_N$. Applying this structural property to the overall Hessian matrix, we obtain the exact upper bound
		\begin{align}
			\nabla_{\bm{x}}^2 \mathrm{SVC}(\bm{x}^{(i)}, \bm{\tilde{r}}_l) \preceq \frac{2\kappa^2 \tilde{u}_l^2}{N^2} (N \bm{I}_N) = \frac{2\kappa^2 \tilde{u}_l^2}{N} \bm{I}_N.
		\end{align}
		Thus, we can extract the global Lipschitz constant $Z_{x,l}$ as the coefficient of the identity matrix, i.e.,
		$Z_{x,l} = \frac{2\kappa^2 \tilde{u}_l^2}{N},$
		which is independent of MA positions. This completes the proof.
	\end{proof}
	Theorem \ref{theorem2} provides a compact and position-independent Lipschitz constant. This enables us to safely construct the quadratic penalty term in the SCA surrogate function.

	For \eqref{varP5}, we rewrite it into a matrix form
	\begin{align}
		\var(\bm{x}) = \frac{1}{N} \bm{x}^T \bm{J} \bm{x},
	\end{align}
	where $\bm{J} = \bm{I}_N - \frac{1}{N}\bm{1}_N\bm{1}_N^T$. Since $\bm{J}$ is positive semi-definite, it is a convex function with respect to $\bm{x}$. Therefore, its first-order Taylor expansion at $\bm{x}^{(i)}$ serves as the global lower bound function, yielding
	\begin{align} \label{Varxlow}
		\var(\bm{x}) &\ge \underline{\var}(\bm{x}; \bm{x}^{(i)}) \triangleq \var(\bm{x}^{(i)}) + \nabla_{\bm{x}} \var(\bm{x}^{(i)})^T (\bm{x} - \bm{x}^{(i)}) \nonumber\\ &= \frac{1}{N} (\bm{x}^{(i)})^T \bm{J} \bm{x}^{(i)} + \frac{2}{N} (\bm{J} \bm{x}^{(i)})^T (\bm{x} - \bm{x}^{(i)}),  
	\end{align}
	where the second equality holds because $\nabla_{\bm{x}} \var(\bm{x}) = \frac{2}{N} \bm{J} \bm{x}$.
	
	For \eqref{JianP4}, we define
	\begin{align}
		d_{n,m}(\bm{x}) = \sqrt{(x_n - x_m)^2 + (y_n^{(i)} - y_m^{(i)})^2}.
	\end{align}
	Obviously, $d_{n,m}(\bm{x}) $ is a convex function with respect to $\bm{x}$, and its first-order Taylor expansion serves as the global lower bound function \eqref{dmnx}.
	\begin{figure*}
		\vspace{-0.5em}
		\begin{align}\label{dmnx}
			d_{n,m}(\bm{x}) &\ge d_{n,m}(\bm{x}^{(i)}) + \frac{\partial d_{n,m}(\bm{x}^{(i)})}{\partial x_n}  (x_n - x_n^{(i)}) + \frac{\partial d_{n,m}(\bm{x}^{(i)})}{\partial x_m}  (x_m - x_m^{(i)}) = d_{n,m}(\bm{x}^{(i)}) + \frac{x_n^{(i)} - x_m^{(i)}}{d_{n,m}(\bm{x}^{(i)})} \Big[(x_n - x_m) - (x_n^{(i)} - x_m^{(i)}) \Big]\nonumber \\ &= \frac{x_n^{(i)} - x_m^{(i)}}{d_{n,m}(\bm{x}^{(i)})} (x_n - x_m) + \Big[ d_{n,m}(\bm{x}^{(i)}) - \frac{(x_n^{(i)} - x_m^{(i)})^2}{d_{n,m}(\bm{x}^{(i)})} \Big]\triangleq \underline{d}_{n,m}(\bm{x}, \bm{x}^{(i)}) .
		\end{align}
		\hrule
		\vspace{-0.5em}
	\end{figure*}
	Then, the optimization problem when optimizing $\bm{x}$ with fixed $\bm{y}$ can be formulated as
	\begin{subequations}\label{P6}
		\begin{align}
			\mathrm{(P6)}\enspace\min_{\bm{x}, \tau}\enspace &\tau \\
			\mathrm{s.t.} \enspace &\overline{\mathrm{SVC}}(\bm{x}; \bm{x}^{(i)}, \bm{\tilde{r}}_l) \le \tau, l= 1, \ldots, L \label{SVCP6}\\
			&\underline{\var}(\bm{x}; \bm{x}^{(i)})\ge \delta_x \label{varP6}\\
			& 0\le x_n\le A_x,\enspace n=1,\ldots,N \\
			&\underline{d}_{n,m}(\bm{x};\bm{x}^{(i)}) \ge d_{\min}, n\neq  m. \label{JianP6}
		\end{align}
	\end{subequations}
	It is worth noting that problem (P6) is convex and we can solve it using the CVX toolbox in MATLAB \cite{mosek2019}.

	\subsection{Optimizing \texorpdfstring{$\bm{y}$}{w} with Fixed \texorpdfstring{$\bm{x}$}{w}}
	When optimizing $\bm{y}$ with fixed $\bm{x}$, the optimization problem can be simplified as
	\begin{subequations}
		\begin{align}
			\mathrm{(P7)}\enspace\min_{\bm{y}, \tau}\enspace &\tau \\
			\mathrm{s.t.} \enspace &\mathrm{SVC}(\bm{y}, \bm{\tilde{r}}_l) \le \tau, l= 1, \ldots, L \label{SVCP7}\\
			&\var(\bm{y})\ge \delta_y \label{varP7}\\
			& \eqref{Ay}, \eqref{JianP4}.
		\end{align}
	\end{subequations}
	It is worth noting that problems (P7) and (P4) exhibit rotational symmetry with respect to $\bm{x}$ and $\bm{y}$. That is, replacing $\bm{y}$ in (P7) with $\bm{x}$ yields problem (P4). Therefore, we adopt a method similar to optimizing $\bm{x}$ with fixed $\bm{y}$. At the $i$-th iteration, deriving the convex lower bound for \eqref{SVCP7} yields
	\begin{align}
		\mathrm{SVC}(\bm{y}, \bm{\tilde{r}}_l)&\le \mathrm{SVC}(\bm{y}^{(i)}, \bm{\tilde{r}}_l)+\nabla_{\bm{y}} \mathrm{SVC}(\bm{y}^{(i)}, \bm{\tilde{r}}_l)^T(\bm{y}- \bm{y}^{(i)})\nonumber\\
		&+\frac{Z_{y,l}}{2}||\bm{y}-\bm{y}^{(i)}||_2^2 \triangleq \overline{\mathrm{SVC}}(\bm{y};\bm{y}^{(i)}, \bm{\tilde{r}}_l),
	\end{align}
where $\nabla_{\bm{y}} \mathrm{SVC}(\bm{y}^{(i)}, \bm{\tilde{r}}_l)$ denotes the first-order gradient of SVC with respect to $\bm{y}$ and can be expressed as
	\begin{align}
		\big[\nabla_{\bm{y}} \mathrm{SVC}(\bm{y}^{(i)}, \bm{\tilde{r}}_l)\big]_{n} = 
		-\frac{2\kappa \tilde{v}_l}{N^2} \sum_{m=1}^{N} \sin (\varphi_{n,m,l}^{(i)}).
	\end{align}
	$Z_{y,l}$ represents the Lipschitz with respect to $v$ and satisfy 
	\begin{align}
		Z_{y,l} \bm{I}_N\succeq \nabla_{\bm{y}}^2 \mathrm{SVC}(\bm{y}^{(i)}, \bm{\tilde{r}}_l),
	\end{align}
	where $\nabla_{\bm{y}}^2 \mathrm{SVC}(\bm{y}^{(i)}, \bm{\tilde{r}}_l)$ denotes the Hessian matrix of $\mathrm{SVC}(\bm{y}^{(i)}, \bm{\tilde{r}}_l)$ with respect to $\bm{y}$ and can be given by 
	\begin{align}
		\nabla_{\bm{y}}^2 \mathrm{SVC}(\bm{y}^{(i)}, \bm{\tilde{r}}_l) = \frac{2\kappa^2 \tilde{v}_l^2}{N^2} \Big[ \bm{B}_l^{(i)} - \diag\big(\bm{B}_l^{(i)}\bm{1}_N\big) \Big].
	\end{align}
	By exploit Theorem \ref{theorem2}, we can obtain the  Lipschitz constant with respect to $\nabla_{\bm{y}}^2 \mathrm{SVC}(\bm{y}^{(i)}, \bm{\tilde{r}}_l)$, i.e., $Z_{y,l} = \frac{2\kappa^2 \tilde{v}_l^2}{N}$.
	
	Similar to \eqref{Varxlow},
	we derive the convex lower bound for \eqref{varP7} as
	\begin{align}
		\!\!\!\!\!\var(\bm{y}) \!\ge\! \underline{\var}(\bm{y}; \bm{y}^{(i)}\!) 
		\!\triangleq \! \frac{1}{N} {\bm{y}^{(i)}}^T \bm{J} \bm{y}^{(i)} \!+\! \frac{2}{N} (\bm{J} \bm{y}^{(i)})^T (\bm{y} \!-\! \bm{y}^{(i)}).  
	\end{align}
	Similar to \eqref{dmnx}, we derive the convex lower bound for \eqref{JianP4}
	\begin{align}
		d_{n,m}(\bm{y}) &\ge \underline{d}_{n,m}(\bm{y}, \bm{y}^{(i)})
		\\
		&\triangleq  d_{n,m}(\bm{y}^{(i)}) + \frac{y_n^{(i)} - y_m^{(i)}}{d_{n,m}(\bm{y}^{(i)})} \Big[(y_n - y_m) - (y_n^{(i)} - y_m^{(i)}) \Big]\nonumber.
	\end{align}
	Thus, at the $i$-th iteration, the optimization problem can be expressed as
	\begin{subequations}\label{P8}
		\begin{align}
			\mathrm{(P8)}\enspace\min_{\bm{y}, \tau}\enspace &\tau \\
			\mathrm{s.t.} \enspace &\overline{\mathrm{SVC}}(\bm{y}; \bm{y}^{(i)}, \bm{\tilde{r}}_l) \le \tau, l= 1, \ldots, L \label{SVCP8}\\
			&\underline{\var}(\bm{y}; \bm{y}^{(i)})\ge \delta_y \label{varP8}\\
			& 0\le y_n\le A_y,\enspace n=1,\ldots,N  \\
			&\underline{d}_{n,m}(\bm{y};\bm{y}^{(i)}) \ge d_{\min}, n\neq  m, \label{JianP8}
		\end{align}
	\end{subequations}
	where $\delta_y$ represents the threshold, and (P8) is a convex problem that can also be solved using the CVX toolbox in MATLAB \cite{mosek2019}.
	
	To this end, we have presented the complete SCA-Based AO algorithm, whose specific steps are summarized in \ref{Algorithm 1}. Since the constructed surrogate functions provide strict global bounds for the original non-convex terms, the objective value monotonically decreases at each iteration. Given that the objective value is bounded from below, the iterative alternating optimization process is guaranteed to strictly converge to a stationary point. The computational complexity predominantly originates from solving the convex subproblems using the interior point method. Each subproblem involves $N+1$ variables and approximately $L + N^2$ linear and quadratic constraints. Consequently, the computational complexity per iteration is bounded by $\mathcal{O}((L + N^2)^{1.5} N^2)$. The overall computational complexity can be ultimately expressed as $\mathcal{O}(I_{\mathrm{iter}} (L + N^2)^{1.5} N^2)$, where $I_{\mathrm{iter}}$ represents the total number of iterations required for convergence.
	
	\begin{algorithm}[t]
		\renewcommand{\algorithmicrequire}{\textbf{Input }}
		\renewcommand{\algorithmicensure}{\textbf{Output }}
		\caption{Proposed SCA-Based AO Algorithm}
		\label{Algorithm 1}
		\begin{algorithmic}[1]
			\Require
			Number of MAs $N$, number of targets $K$,
			thresholds $\delta_x$ and $\delta_y$,
			movable region $A_x\times A_y$,
			sampling points $\{\bm{\tilde{r}}_l\}_{l=1}^{L}$,
			maximum iteration number $I_{\mathrm{iter}}$,
			and convergence tolerance $\epsilon$.
			
			\State \textbf{Initialize:}
			Set feasible initial MA positions
			$\bm{x}\leftarrow\bm{x}^{(0)}$
			and $\bm{y}\leftarrow\bm{y}^{(0)}$.
			Compute
			$\tau^{(0)}\leftarrow
			\max_{l=1,\ldots,L}
			\mathrm{SVC}(\bm{p}^{(0)},\bm{\tilde{r}}_l)$.
			
			\For{$i=0:I_{\mathrm{iter}}-1$}
			
			\State \textbf{Stage 1: Optimize $\bm{x}$ with fixed $\bm{y}$}
			\State Obtain $\bm{x}^{(i+1)}$
			by solving problem (P6).
			\State $\bm{x}\leftarrow\bm{x}^{(i+1)}$.
			
			\State \textbf{Stage 2: Optimize $\bm{y}$ with fixed $\bm{x}$}
			\State Obtain $\bm{y}^{(i+1)}$
			by solving problem (P8).
			\State $\bm{y}\leftarrow\bm{y}^{(i+1)}$.
			
			\State Update
			$\bm{p}^{(i+1)}
			\leftarrow
			[(\bm{x}^{(i+1)})^T,
			(\bm{y}^{(i+1)})^T]^T$.
			
			\State Compute
			$\tau^{(i+1)}
			\leftarrow
			\max_{l=1,\ldots,L}
			\mathrm{SVC}(\bm{p}^{(i+1)},\bm{\tilde{r}}_l)$.
			
			\If{$|\tau^{(i+1)}-\tau^{(i)}|\leq\epsilon$}
			\State \textbf{break}
			\EndIf
			
			\EndFor
			
			\Ensure
			Optimized MA positions
			$\bm{x}^{\star}\leftarrow\bm{x}^{(i+1)}$
			and
			$\bm{y}^{\star}\leftarrow\bm{y}^{(i+1)}$.
		\end{algorithmic}
	\end{algorithm}
	
	\section{Simulation Results}
	In this section, we evaluate the performance of the proposed algorithm through two comprehensive groups of numerical simulations. In the first group, the target directions are randomly generated initially but remain fixed across all Monte Carlo realizations. Under this setup, we illustrate the optimized 2D antenna positions and the corresponding MUSIC pseudo-spectra for various array configurations. Furthermore, we evaluate the RMSE and the LB of different arrays with respect to the varying SNR, the number of antennas $N$, the size of the movable region $A_x \times A_y$, and the number of targets $K$. In the second group of simulations, the target directions are dynamically and independently generated for each Monte Carlo trial. Based on this dynamic setting, we investigate the estimation success rate, alongside the RMSE and LB of different arrays, as functions of the varying SNR, $N$, $A_x \times A_y$, and $K$.
	The parameters are denoted as follows. The RMSE (deg) is given by
	\begin{align}
		\mathrm{RMSE} = \sqrt{\frac{1}{QK}\sum_{q=1}^{Q}\sum_{k=1}^{K}\big[(\theta_k-\hat{\theta}_{q,k})^2 + (\phi_k-\hat{\phi}_{q,k})^2 \big]},
	\end{align}
	where $\hat{\theta}_{q,k}$ and $\hat{\phi}_{q,k}$ represent the estimated azimuth and elevation angles at the $q$-th Monte Carlo simulation of the $k$-th target. $Q$ denotes the total number of the Monte Carlo simulations. To resolve the permutation ambiguity, the estimated
	DOAs are matched to the true target directions by
	minimizing the sum of squared angular errors over
	all possible one-to-one assignments. The matched
	estimates are then used to calculate the RMSE and
	estimation success rate.
	The SNR (dB) is given by $\mathrm{SNR} = 10\lg \frac{\sigma_k^2}{\sigma_n^2}$.
	Without loss of generality, we set $Q = 500$ and the signal power as $\sigma_k^2 =1$ for $k =1,\ldots,K$. The directions of the targets are generated randomly by
		$\theta_k \sim \mathcal{U}\big(\frac{\pi}{6}, \frac{5\pi}{6}\big)$ and $
		\phi_k \sim \mathcal{U}\big(\frac{\pi}{6}, \frac{5\pi}{6}\big)$.
	Thus, $u_{\max} = v_{\max} = \frac{\sqrt{3}}{2}$. To effectively protect the main lobe and suppress the grating lobes, we set $\gamma_u = \gamma_v = 0.2$. The convergence threshold in Algorithm \ref{Algorithm 1} is set as $\epsilon = 10^{-4}$.
	To simulate realistic scenarios, a minimum spatial separation must be maintained between any two distinct target directions. Specifically, they are required to satisfy the following distance constraint in the directional space
	\begin{align}
		\sqrt{(\theta_k - \theta_{k'})^2 + (\phi_k - \phi_{k'})^2} \ge \Delta, 
	\end{align}
	where the minimum separation threshold $\Delta$ is set to $10^{\circ}$.
	Unless otherwise specified, the receiver is equipped with $N = 25$ MAs operating under a movable region of $A_x = A_y = 8 \lambda$. Furthermore, we collect $T = 20$ snapshots for each observation under an SNR of 10 $\mathrm{dB}$ for estimating  $K = 4$ targets. For performance comparison, we consider the following baseline schemes:
	\begin{itemize}
		\item \textbf{PMA}: The MA positions are optimized by the proposed algorithm.
		\item \textbf{CMA}: The MA positions are optimized by the conventional algorithm in \cite{maMovableAntennaEnhanced2024} without considering multi-target coupling.
		\item \textbf{URA}: The conventional uniform rectangular array (URA) with half-wavelength spacing between adjacent antennas as shown in \Cref{Fig2D}.
		\item \textbf{USA}: The uniform sparse array (USA) as shown in \Cref{Fig2D} with inter-element spacings of $d_x = \frac{A_x}{\sqrt{N}-1}$ and $d_y = \frac{A_y}{\sqrt{N}-1}$ in the $x$- and $y$-axis, respectively.
	\end{itemize}
	
	\begin{figure}[t]
		\vspace{-0.5em}
		\centering
		\includegraphics[width=0.75\linewidth]{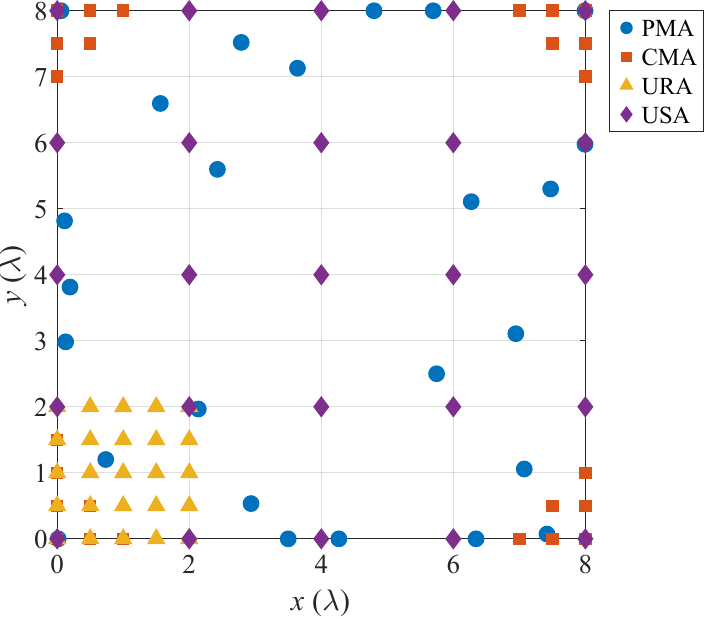}
		\caption{The antenna positions of different array configurations.}
		\label{Fig2D}
		\vspace{-0.5em}
	\end{figure}
	
	\subsection{Simulations with Fixed Targets}
	\begin{figure*}[t]
		\vspace{-0.5em}
		\centering
		\subfloat[PMA]{
			\includegraphics[width=0.22\textwidth]{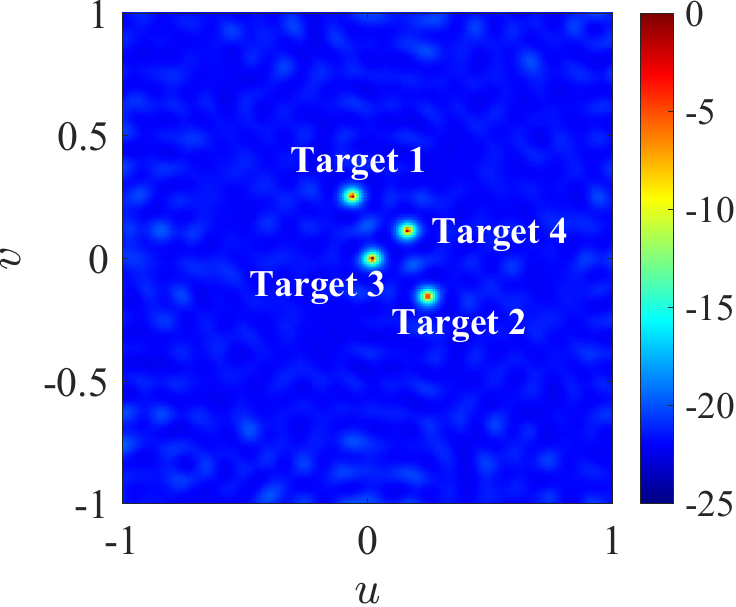}
			\label{PMA}
		}
		\hfill
		\subfloat[CMA]{
			\includegraphics[width=0.22\textwidth]{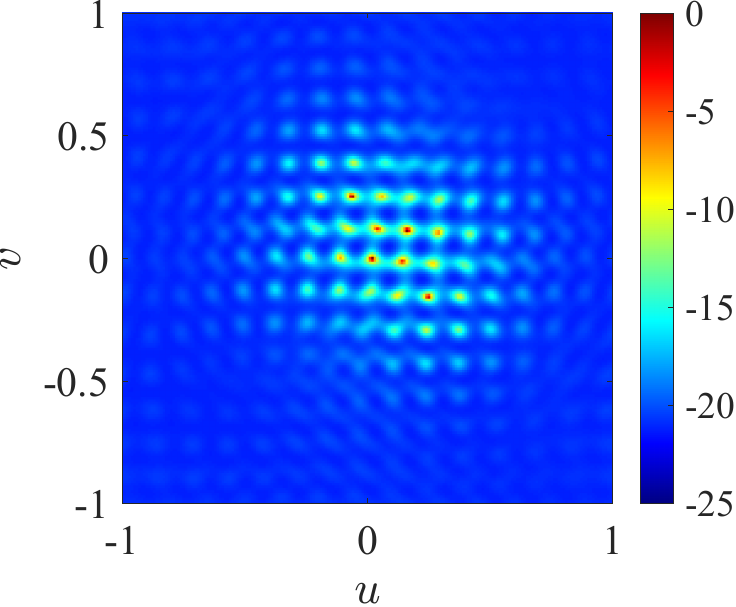}
			\label{CMA}
		}
		\hfill
		\subfloat[URA]{
			\includegraphics[width=0.22\textwidth]{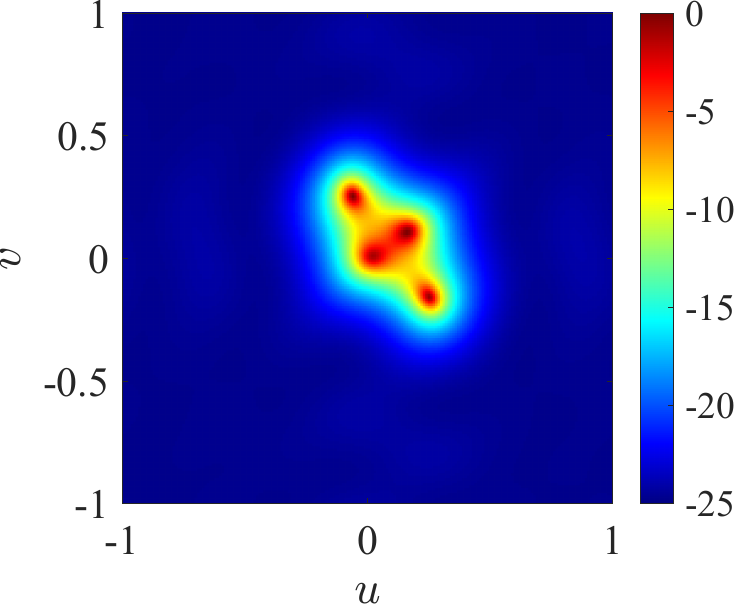}
			\label{URA}
		}
		\hfill
		\subfloat[USA]{
			\includegraphics[width=0.22\textwidth]{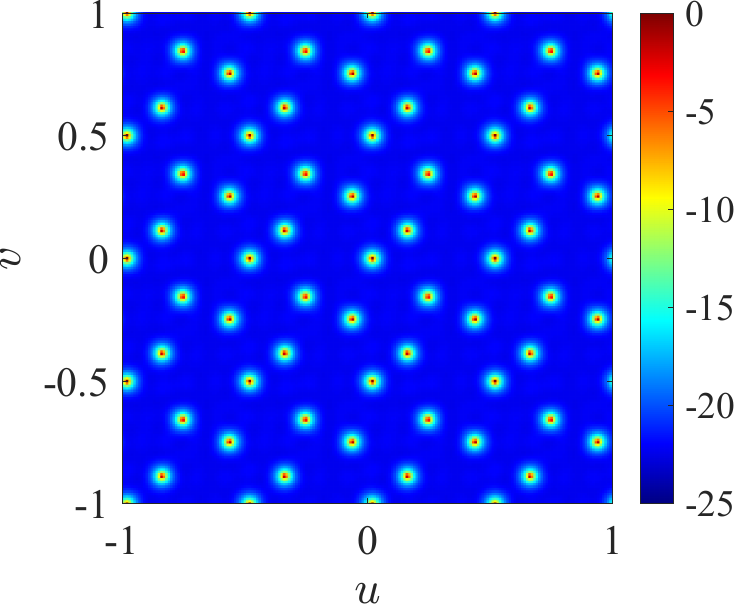}
			\label{USA}
		}
		\vspace{-0.5em}
		\caption{MUSIC spatial spectrum of different array configurations.}
		\label{2d_spectrum_1x4}
		\vspace{-0.5em}
	\end{figure*}
	
	\begin{figure}[t]
		\centering
		\includegraphics[width=0.75\linewidth]{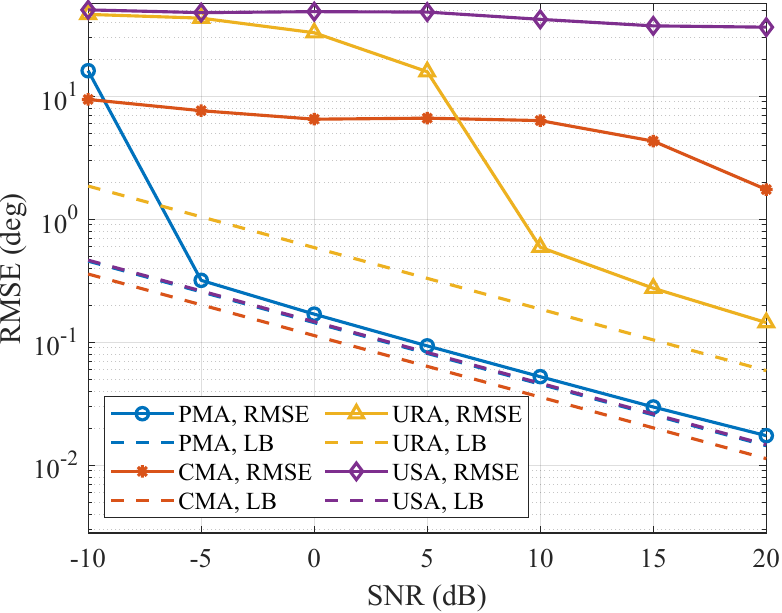}
		\caption{RMSE versus SNR.}
		\label{RMSE_SNR_2D}
		\vspace{-0.5em}
	\end{figure}
	
%
	\begin{figure}[t]
		\centering
		\includegraphics[width=0.75\linewidth]{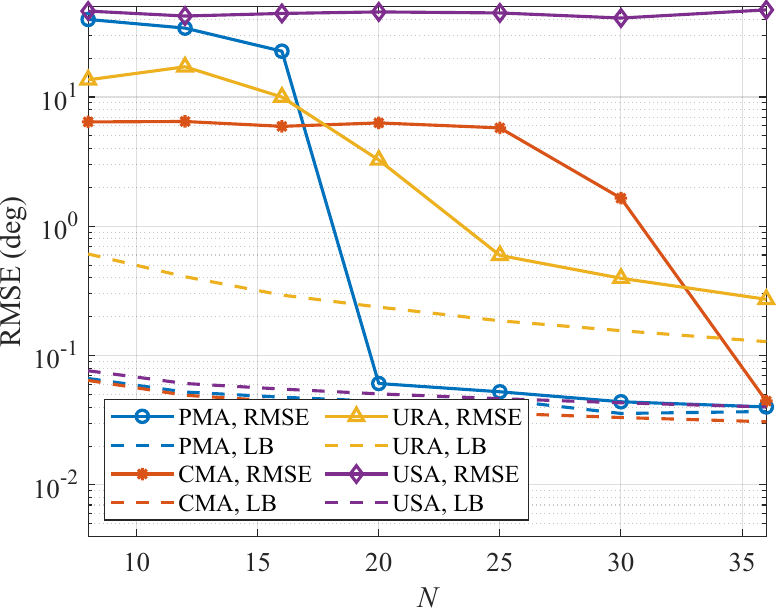}
		\caption{RMSE versus the number of MAs.}
		\label{RMSE_N}
		\vspace{-0.5em}
	\end{figure}
	
	\begin{figure}[t]
		\centering
		\includegraphics[width=0.78\linewidth]{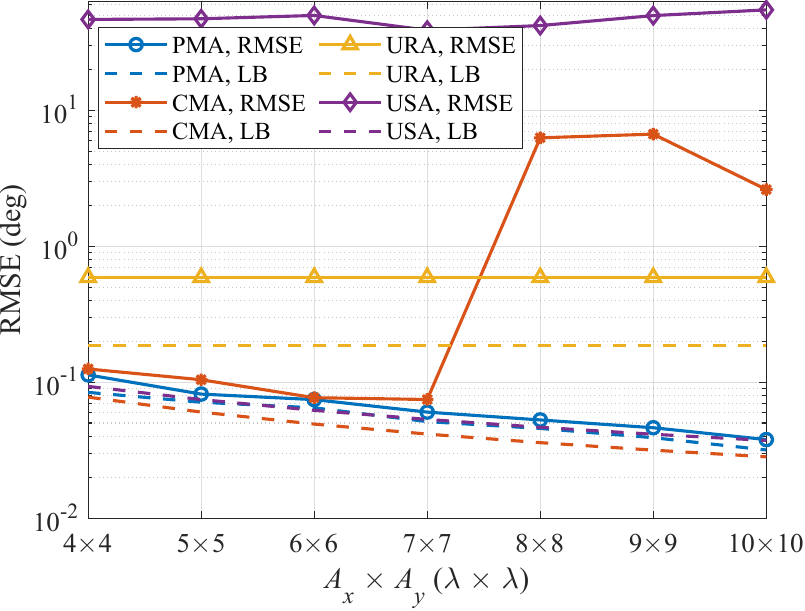}
		\caption{RMSE versus the size of the movable region.}
		\label{RMSE_L}
		\vspace{-0.5em}
	\end{figure}
	
		\begin{figure}[t]
		\centering
		\includegraphics[width=0.75\linewidth]{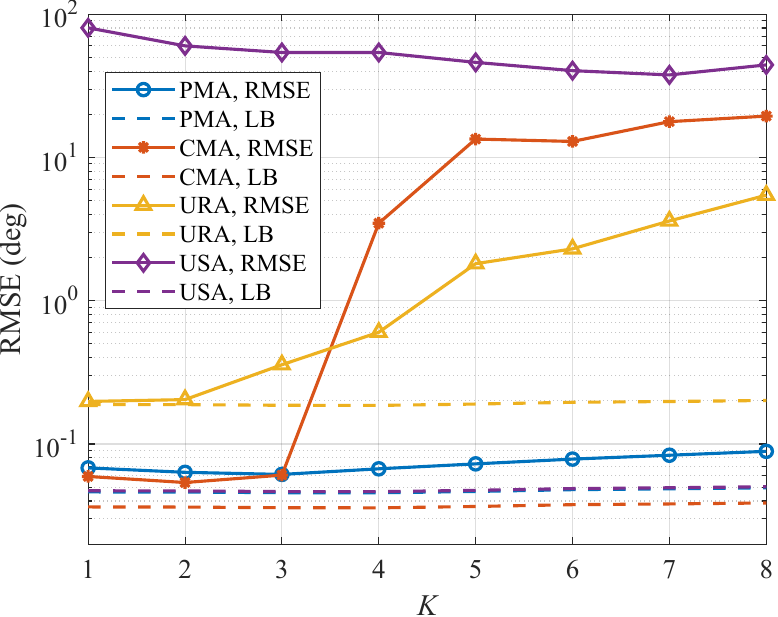}
		\caption{RMSE versus the number of targets.}
		\label{RMSE_K}
		\vspace{-0.5em}
	\end{figure}
	In this subsection, we focus on the simulations with the fixed random targets to demonstrate the superiority of the proposed algorithm. The directions of the targets are randomly generated as $\bm{\theta}  = [75.45, 98.95, 90.13, 83.52]^T,
	\bm{\phi}  = [93.78, 75.59, 88.86, 80.61]^T.$

	\Cref{Fig2D} illustrates the MA positions for different array configurations. For the conventional URA, the antenna spacing along both axes is fixed at $\lambda/2$. Consequently, its total aperture size is restricted to $2\lambda \times 2\lambda$, which occupies merely one-sixteenth of the available movable region. The USA evenly occupies the entire movable region with an enlarged antenna spacing of $2\lambda$ along both axes. As a result of naively minimizing the CRB, the antennas of the CMA are predominantly distributed at the four corners of the movable region, leaving a massive spatial hole in the center. In contrast, the antennas of the proposed PMA are scattered more effectively across the entire movable region, significantly reducing the size of the central spatial hole.
	
	\Cref{2d_spectrum_1x4} depicts the 2D MUSIC spatial spectrum of the different array configurations during a single Monte Carlo trial. As clearly observed in \Cref{PMA}, the spatial spectrum of the proposed PMA yields four distinct and sharp peaks corresponding to the actual targets, accompanied by almost no severe sidelobes. The target directions can be successfully estimated by simply identifying the four highest peaks. Observing \Cref{CMA}, multiple high peaks emerge across the spatial domain. Among them, seven prominent peaks could easily be misinterpreted as potential target directions. This phenomenon occurs because the massive spatial hole in the CMA configuration induces high sidelobes in the SVC, which exacerbates mutual interference among targets and leads to potentially significant estimation errors. In \Cref{URA}, four peaks corresponding to the true target directions are also visible. However, these peaks are considerably less sharp than those of the PMA. Furthermore, the spectrum valley between the third and fourth targets remains excessively high, posing a severe challenge in resolving these two closely spaced targets. Finally, observing the USA configuration, numerous peaks appear periodically across the spatial spectrum. This periodicity arises because the large antenna spacing of the USA generates severe periodic grating lobes in the SVC, making it fundamentally impossible to distinguish the true targets from the spurious ones.
	
	\Cref{RMSE_SNR_2D} illustrates the RMSE and LB of different array configurations versus the SNR. It can be observed that the RMSE of the proposed PMA approaches a low LB at an SNR of -5 dB. On the one hand, the PMA minimizes the sidelobe levels, thereby mitigating mutual interference among multiple targets. On the other hand, the PMA guarantees that the LB does not exceed a predefined threshold, maintaining the RMSE at a comparatively low level. Although the CMA explicitly minimizes the LB, its configuration suffers from a massive spatial hole in the center of the array. This leads to severe mutual interference among targets, making it incapable of successfully estimating the DOAs. Similarly, the USA generates periodic grating lobes, which entirely precludes accurate DOA estimation. For the URA, its failure to fully exploit the available aperture results in an LB significantly higher than those of the PMA and CMA, leading to lower spatial resolution. Consequently, the URA fails to estimate the DOAs at low SNRs and only approaches the LB when the SNR exceeds 10 dB.
	
	\Cref{RMSE_N} demonstrates the RMSE of different arrays as a function of the number of antennas $N$. As the number of antennas increases, the RMSEs of the PMA, CMA, and URA consistently decrease. The PMA can approach the CRB with as few as $ N = 20$, whereas the CMA requires $N = 36$ to achieve similar convergence. On the one hand, an increased number of antennas inherently enhances the estimation performance. On the other hand, for a fixed movable region, a larger number of antennas effectively reduces the size of spatial holes and subsequent inter-target interference, which further contributes to the performance gain. The URA only approaches the CRB when $N = 25$, because its limited physical aperture at smaller values of $N$ yields insufficient target resolution. The USA remains incapable of estimating DOAs since its antenna spacing is strictly larger than $\lambda/2$. Thus, accurate DOA estimation with the USA becomes feasible only if its spacing decreases and it degenerates into a URA configuration.
	
	\Cref{RMSE_L} displays the RMSE of different arrays as the size of the movable region varies. As the movable antenna region expands, the LBs of the PMA, CMA, and USA decrease accordingly. Since the deterministic antenna positions of the URA are strictly independent of the movable region size, its LB and RMSE remain constant. The USA completely fails to estimate the DOAs due to severe grating lobes. The RMSE of the proposed PMA consistently approaches the CRB across all evaluated region sizes. Conversely, the CMA can only estimate the DOAs successfully within a relatively small movable region. This limitation occurs because an enlarged movable region exacerbates the massive spatial holes in the CMA configuration, thereby intensifying the mutual interference among targets.
	
	\Cref{RMSE_K} presents the RMSE of different arrays as a function of the number of targets $K$. Except for the USA, which fails to estimate DOAs regardless of the target count, the estimation performance of the PMA, CMA, and URA inevitably deteriorates as the number of targets increases. The PMA consistently approaches the LB across all evaluated target quantities. The CMA approaches the LB only when the number of targets is small, and its performance degrades sharply when $K > 3$, rendering it incapable of estimating the DOAs. This sharp deterioration is attributed to the fact that a larger number of targets significantly amplifies the mutual interference induced by the spatial holes in the CMA configuration. Similarly, the URA only approaches the LB under a small target count and experiences severe performance degradation when $K > 5$. This limitation arises because the URA possesses a restricted physical aperture, which inherently yields a low angular resolution that is insufficient to resolve a larger number of target directions.
	\vspace{-0.5em}
	\subsection{Simulations with Random Targets}
	
	\begin{figure}[t]
		\vspace{-0.5em}
		\centering
		\includegraphics[width=0.75\linewidth]{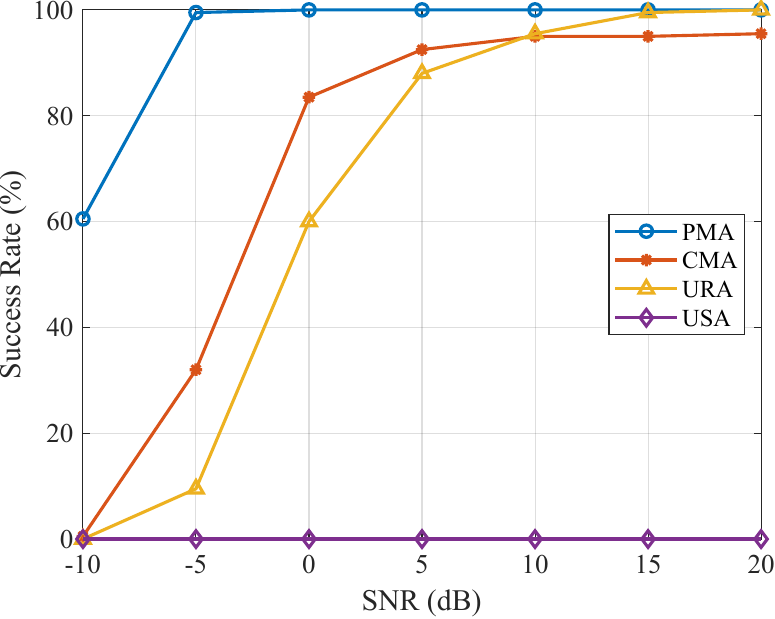}
		\caption{Success rate versus SNR.}
		\label{Rate_SNR_2D}
		\vspace{-0.5em}
	\end{figure}
	\begin{figure}[t]
		\vspace{-0.5em}
		\centering
		\includegraphics[width=0.75\linewidth]{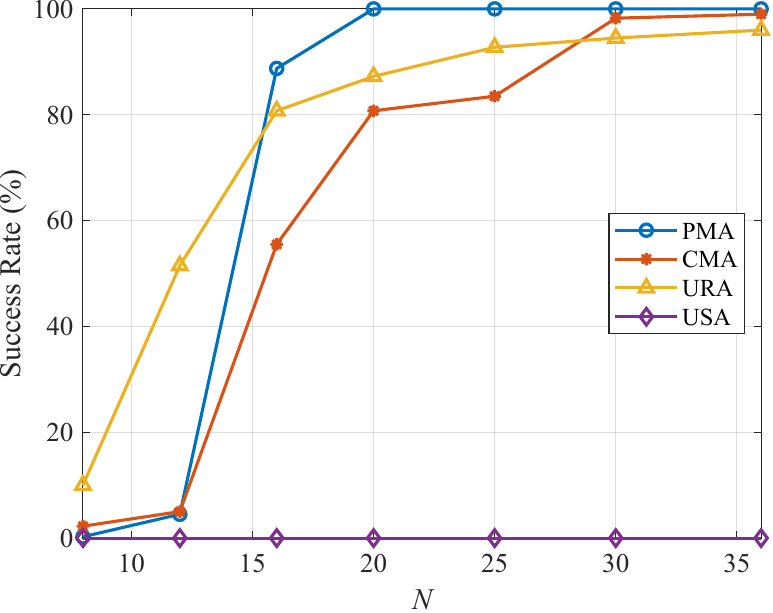}
		\caption{Success rate versus the number of MAs.}
		\label{Rate_N}
		\vspace{-0.5em}
	\end{figure}
	
	\begin{figure}[t]
		\vspace{-0.5em}
		\centering
		\includegraphics[width=0.75\linewidth]{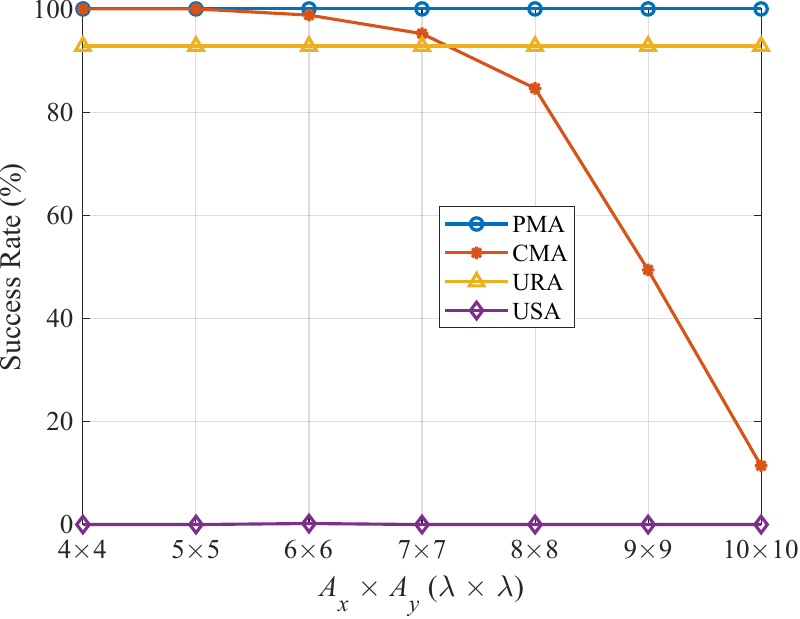}
		\caption{Success rate versus the size of the movable region.}
		\label{Rate_L}
		\vspace{-0.5em}
	\end{figure}
	\begin{figure}[t]
		\vspace{-0.5em}
		\centering
		\includegraphics[width=0.75\linewidth]{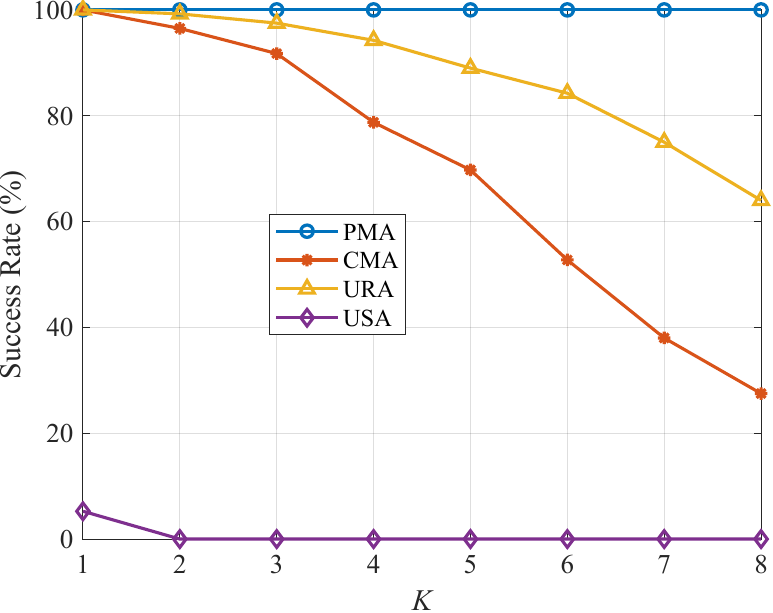}
		\caption{Success rate versus the number of the targets.}
		\label{Rate_K}
		\vspace{-0.5em}
	\end{figure}
	In this subsection, we evaluate the success rate of different arrays with respect to the varying SNR, $N$, $A_x \times A_y$, and $K$ with random targets in each trial.
	
	To quantitatively evaluate the reliability and robustness of the DOA estimation across multiple Monte Carlo trials, we introduce a success indicator function for the $q$-th trial. Specifically, an estimation is considered successful if the average Euclidean distance between the estimated and true angles across all $K$ targets falls below a predefined tolerance threshold. The success indicator function $\mathcal{I}_{q}$ is thus defined as
	\begin{align}
		\!\!\mathcal{I}_{q} =
		\begin{cases}
			1, &\!\!\!\!\!\! \sqrt{ \frac{1}{K} \sum_{k=1}^{K} \big[ (\theta_k - \hat{\theta}_{q, k})^2 + (\phi_k - \hat{\phi}_{q, k})^2 \big] } \le \beta, \\
			0, &\!\!\!\!\! \text{otherwise},
		\end{cases}
	\end{align}
	where $\beta$ denotes the success threshold, which is empirically set to $1^{\circ}$. Based on this indicator, the overall estimation success rate over $Q$ independent Monte Carlo trials is calculated by
	\begin{align}
		\text{Success Rate} = \Big( \frac{1}{Q} \sum_{q=1}^{Q} \mathcal{I}_{q}  \Big) \times 100\%.
	\end{align}
	
	\Cref{Rate_SNR_2D} illustrates the estimation success rate of different arrays versus the SNR. The proposed PMA demonstrates remarkable robustness, rapidly achieving a 100\% success rate at an exceptionally low SNR of -5 dB. This superior performance validates that minimizing the sidelobe levels effectively prevents noise-induced false peaks, ensuring reliable target extraction in low-SNR regimes. The CMA requires a higher SNR of 5 dB to achieve a comparable success rate, as its massive spatial holes cause high sidelobes that are easily corrupted by noise. The URA, constrained by its limited physical aperture, demands an even higher SNR of 20 dB to reliably resolve the targets. As expected, the USA completely fails, maintaining a 0\% success rate across all SNR conditions due to severe grating lobes.
	
	\Cref{Rate_N} depicts the success rate of various configurations as a function of the number of antennas $N$. As $N$ increases, the success rates of the PMA, CMA, and URA generally improve. The PMA rapidly converges to a 100\% success rate with only $N=20$ antennas. This highlights that an optimally distributed movable array can resolve targets efficiently with a significantly reduced hardware footprint. In contrast, the CMA requires at least $N=30$ antennas to reach a near 100\% success rate, since a larger number of elements is needed to partially fill its spatial holes and reduce interference. The URA experiences a slower ascent, achieving roughly 95\% at $N=35$ due to its strictly confined aperture. Meanwhile, the USA remains completely incapable of estimating the DOAs, exhibiting a 0\% success rate regardless of $N$.
	
	\Cref{Rate_L} shows the success rate evaluated against the size of the movable region $A_x \times A_y$. The PMA consistently maintains a perfect 100\% success rate across all evaluated region sizes, proving the high adaptability of the proposed algorithm in strategically suppressing sidelobes regardless of the deployment area. Since the URA's geometry is fixed and independent of the movable region, its success rate remains constant at approximately 95\%. Interestingly, the CMA achieves a high success rate only in smaller movable regions (e.g., $4\times4$). As the region expands to $10\times10$, the success rate of the CMA plummets to roughly 15\%. This catastrophic degradation occurs because an enlarged region exacerbates the boundary-clustering effect, expanding the central spatial hole and intensifying mutual interference. The USA remains completely ineffective throughout. 
	
	\Cref{Rate_K} illustrates the estimation success rate with respect to the number of targets $K$. The proposed PMA exhibits exceptional robustness, sustaining a 100\% success rate even when the number of co-existing targets increases to $K=8$. This directly verifies that minimizing the sidelobe levels effectively mitigates multi-target mutual masking. The CMA's performance drops sharply as $K$ increases, falling to approximately 30\% at $K=8$. This rapid decline confirms that the high sidelobes generated by the CMA's spatial holes are highly susceptible to severe inter-target interference in dense target scenarios. The URA configuration also degrades significantly, dropping to a success rate of about 55\% at $K=8$, primarily because its restricted physical aperture lacks the spatial resolution required to separate a large number of closely spaced targets. Finally, the USA configuration fails entirely to resolve multiple targets, dropping to a 0\% success rate.
	\vspace{-0.5em}
	\section{Conclusion}
	In this paper, we proposed a multi-target wireless sensing scheme based on an MA array to achieve high-precision and robust target parameter estimation. We introduced two complementary metrics, namely the LB of the CRB and the SVC, to characterize the estimation accuracy and sidelobe interference, respectively. A closed-form LB was derived by exploiting the partial order relationship of projection matrices, which decouples the unknown target directions from the MA positions. Based on these two metrics, we formulated an MA position optimization problem to minimize the SVC while satisfying a prescribed LB constraint. The continuous SVC objective and the LB constraint were further transformed into tractable discrete and spatial variance constraints, respectively. Subsequently, an SCA-based AO algorithm was developed to iteratively optimize the MA positions, with its monotonic convergence guaranteed by a position-independent Lipschitz constant. Numerical results demonstrated that the proposed MA configuration effectively suppresses SVC sidelobes while achieving superior estimation accuracy and sensing robustness over conventional fixed-position arrays.
	
	\appendices
	\section{CRB Derivation}\label{AppendixA}
	In this appendix, we aim to derive the CRB of the multi-target sensing systems with MAs.
	To facilitate the derivation, we vectorize both sides of \eqref{YY} as
	\begin{align}\label{yvec}
		\bm{\tilde{y}}=  (\bm{I}_{T} \otimes \bm{A})\bm{\tilde{s}}
		+\bm{\tilde{n}}= \bm{\mu}+\bm{\tilde{n}},
	\end{align}
	where the equality utilizes the property of the Kronecker product $\operatorname{vec}(\bm{ABC}) = (\bm{C}^T \otimes \bm{A})\operatorname{vec}(\bm{B})$ with $ \bm{\tilde{s}} = {\operatorname{vec}(\bm{S})} $ and ${\operatorname{vec}(\bm{N})}=\bm{\tilde{n}}$. Equation \eqref{yvec} can be regarded as a Gaussian  distribution with mean $\bm{\mu}$ and variance $\sigma_n^2$. The unknown parameters to be estimated are the directions of the $K$ targets and the signal $\bm{S}$. Thus, we define the unknown parameter vector as
	\begin{align}
		\bm{\psi}=[\bm{\omega}^T, \bm{\tilde{s}}^T_{R}, \bm{\tilde{s}}^T_{I}  ]^T\in \mathbb{R}^{(2K+2KT)  \times 1},
	\end{align}
	where $\bm{\tilde{s}}_{R} = \mathfrak{R}\{\bm{\tilde{s}}^T\}$, and $\bm{\tilde{s}}_{I} = \mathfrak{I}\{\bm{\tilde{s}}^T\}$. Therefore, the FIM matrix $\bm{F}\in \mathbb{R}^{(2K+2KT) \times (2K+2KT) }$ can be partitioned into the following block matrix
	\begin{align}
			\bm{F} &= \begin{bmatrix}
				\bm{F}_{\bm{\omega}\bm{\omega}} & \bm{F}_{\bm{\omega}\bm{\tilde{s}}_R} & \bm{F}_{\bm{\omega}\bm{\tilde{s}}_I} \\
				\bm{F}^T_{\bm{\omega}\bm{\tilde{s}}_R} & \bm{F}_{\bm{\tilde{s}}_R\bm{\tilde{s}}_R} & \bm{F}_{\bm{\tilde{s}}_R\bm{\tilde{s}}_I} \\
				\bm{F}^T_{\bm{\omega}\bm{\tilde{s}}_I} & \bm{F}^T_{\bm{\tilde{s}}_R\bm{\tilde{s}}_I} & \bm{F}_{\bm{\tilde{s}}_I\bm{\tilde{s}}_I}
			\end{bmatrix}=
			\begin{bmatrix}
				\bm{F}_{\bm{\omega}\bm{\omega}} & \bm{F}_{\bm{\omega}\bm{\tilde{s}}} \\
				\bm{F}^T_{\bm{\omega}\bm{\tilde{s}}} & \bm{F}_{\bm{\tilde{s}}\bm{\tilde{s}}} 
			\end{bmatrix},
	\end{align}
	where the second equality holds true by defining
	\begin{align}
		\bm{F}_{\bm{\omega}\bm{\tilde{s}}} &\triangleq \big[ \bm{F}_{\bm{\omega}\bm{\tilde{s}}_R}, \, \bm{F}_{\bm{\omega}\bm{\tilde{s}}_I} \big], \quad
		\bm{F}_{\bm{\tilde{s}}\bm{\tilde{s}}} \triangleq \begin{bmatrix}
			\bm{F}_{\bm{\tilde{s}}_R\bm{\tilde{s}}_R} & \bm{F}_{\bm{\tilde{s}}_R\bm{\tilde{s}}_I} \\
			\bm{F}^T_{\bm{\tilde{s}}_R\bm{\tilde{s}}_I} & \bm{F}_{\bm{\tilde{s}}_I\bm{\tilde{s}}_I}
		\end{bmatrix}.
	\end{align}
	
	Furthermore, the $(m,n)$-th element of the FIM matrix can be expressed as
	\begin{align}
		[\bm{F}]_{m,n}=	\frac{2}{\sigma_n^2}
		\mathfrak{R}
		\Big\{\Big(\frac{\partial \bm{\mu}}
		{\partial [\bm{\psi}]_m}\Big)^H
		\Big(\frac{\partial \bm{\mu}}
		{\partial [\bm{\psi}]_n}\Big)\Big\}.
	\end{align}
	The first-order derivative of the mean vector $\bm{\mu}$ with respect to the parameter $\bm{\omega}$ can be formulated as
	\begin{align}\label{uw}
		\frac{\partial\bm{\mu}}{\partial\bm{\omega}} 
		=\Big[ \Big(\frac{\partial(\bm{A}\bm{s}(1))}{\partial\bm{\omega}}\Big)^T, \ldots, \Big(\frac{\partial(\bm{A}\bm{s}(T))}{\partial\bm{\omega}}\Big)^T \Big]^T,
	\end{align}
	where
	\begin{align}\label{uwf}
		&\frac{\partial(\bm{A}\bm{s}(t))}{\partial\bm{\omega}}
		=\Big[\frac{\partial(\bm{A}\bm{s}(t))}{\partial\bm{u}},\frac{\partial(\bm{A}\bm{s}(t))}{\partial\bm{v}}\Big] \\
		&=\Big[\frac{\partial\bm{a}_1}{\partial u_1}s_1(t),\ldots,\frac{\partial\bm{a}_K}{\partial u_K}s_K(t),\frac{\partial\bm{a}_1}{\partial v_1}s_1(t),\ldots,\frac{\partial\bm{a}_K}{\partial v_K}s_K(t)\Big]\nonumber \\
		&=\bm{W}[\bm{I}_2\otimes  \diag(\bm{s}(t)) ]=\bm{W}\bm{\tilde{S}}(t)\nonumber,
	\end{align}
	and $
	\bm W=\left[\frac{\partial\bm a_1}{\partial u_1},
	\ldots,\frac{\partial\bm a_K}{\partial u_K},
	\frac{\partial\bm a_1}{\partial v_1},\ldots,
	\frac{\partial\bm a_K}{\partial v_K}\right] \in\mathbb C^{N\times 2K}$.
	Combining \eqref{uw} and \eqref{uwf} yields
	\begin{align}
		\frac{\partial\bm{\mu}}{\partial \bm{\omega}}
		=
		\begin{bmatrix}
			\bm{\tilde{S}}^T(1)\bm{W}^T,
			\ldots,
			\bm{\tilde{S}}^T(T)\bm{W}^T
		\end{bmatrix}^T.
	\end{align}
	Utilizing the mean expansion $\bm{\mu} = (\bm{I}_T \otimes \bm{A})\bm{s}$, the derivatives with respect to the real and imaginary parts are respectively given by
	\begin{align}
		\frac{\partial\bm{\mu}}{\partial\bm{\tilde{s}}_R} = \bm{I}_T \otimes \bm{A}, \enspace \frac{\partial\bm{\mu}}{\partial\bm{\tilde{s}}_I} = j(\bm{I}_T \otimes \bm{A}).
	\end{align}
	Hence, we can obtain \eqref{Fww}-\eqref{Fws},
\begin{figure*}
	\vspace{-1.5em}
	\begin{align}
		\bm{F}_{\bm{\omega\omega}} &= \frac{2}{\sigma_n^2}\mathfrak{R}\Big\{\Big(\frac{\partial \bm{\mu}}{\partial \bm{\omega}}\Big)^H \Big(\frac{\partial \bm{\mu}}{\partial \bm{\omega}}\Big)\Big\} 
		= \frac{2}{\sigma_n^2}\mathfrak{R}\Bigg\{ \big[ \bm{\tilde{S}}^H(1)\bm{W}^H, \ldots, \bm{\tilde{S}}^H(T)\bm{W}^H\big] \begin{bmatrix} \bm{W}\bm{\tilde{S}}(1) \\ \vdots \\ \bm{W}\bm{\tilde{S}}(T) \end{bmatrix} \Bigg\} = \frac{2}{\sigma_n^2}\mathfrak{R}\Big\{ \sum_{t=1}^{T} \bm{\tilde{S}}^H(t)\bm{W}^H\bm{W}\bm{\tilde{S}}(t) \Big\},\label{Fww}\\
		\bm{F}_{\bm{\tilde{s}}_R\bm{\tilde{s}}_R} &= \bm{F}_{\bm{\tilde{s}}_I\bm{\tilde{s}}_I} =  \frac{2}{\sigma_n^2}\mathfrak{R}\Big\{\Big(\frac{\partial \bm{\mu}}{\partial \bm{\tilde{s}}_R}\Big)^H \Big(\frac{\partial \bm{\mu}}{\partial \bm{\tilde{s}}_R}\Big)\Big\} = \frac{2}{\sigma_n^2}\mathfrak{R}\Big\{ (\bm{I}_T \otimes \bm{A})^H (\bm{I}_T \otimes \bm{A}) \Big\} = \frac{2}{\sigma_n^2}\mathfrak{R}\Big\{ \bm{I}_T \otimes (\bm{A}^H\bm{A}) \Big\}, \label{FSS}\\
		\bm{F}_{\bm{\tilde{s}}_R\bm{\tilde{s}}_I} &= -\bm{F}_{\bm{\tilde{s}}_I\bm{\tilde{s}}_R} = \frac{2}{\sigma_n^2}\mathfrak{R}\Big\{\Big(\frac{\partial \bm{\mu}}{\partial \bm{\tilde{s}}_R}\Big)^H \Big(\frac{\partial \bm{\mu}}{\partial \bm{\tilde{s}}_I}\Big)\Big\} = \frac{2}{\sigma_n^2}\mathfrak{R}\Big\{ j(\bm{I}_T \otimes \bm{A})^H (\bm{I}_T \otimes \bm{A}) \Big\} = -\frac{2}{\sigma_n^2}\mathfrak{I}\Big\{ \bm{I}_T \otimes (\bm{A}^H\bm{A}) \Big\}, \label{FSRSI}\\
		\bm{F}_{\bm{\omega}\bm{\tilde{s}}_R} & = 
		\frac{2}{\sigma_n^2} \mathfrak{R} \Big\{ \big[ \bm{\tilde{S}}^H(1)\bm{W}^H, \ldots, \bm{\tilde{S}}^H(T)\bm{W}^H \big] (\bm{I}_T \otimes \bm{A}) \Big\},\quad  \bm{F}_{\bm{\omega}\bm{\tilde{s}}_I} =-	\frac{2}{\sigma_n^2} \mathfrak{I} \Big\{ \big[ \bm{\tilde{S}}^H(1)\bm{W}^H, \ldots, \bm{\tilde{S}}^H(T)\bm{W}^H \big] (\bm{I}_T \otimes \bm{A}) \Big\} \label{Fws},
	\end{align}
	\vspace{-0.5em}
	\hrule
	\vspace{-0.5em}
\end{figure*}
	where the fourth equality in \eqref{FSS} holds true by the identity $(\bm{A}\otimes\bm{B})(\bm{C}\otimes\bm{D}) = (\bm{A}\bm{C})\otimes(\bm{B}\bm{D})$.
	According to the matrix inversion lemma, the inverse of the CRB block matrix with respect to $\bm{\omega}$ can be formulated as \eqref{CRB}.
	\begin{figure*}
		\vspace{-0.5em}
		\begin{align}\label{CRB}
			\bm{CRB}^{-1}(\bm{\omega}) 
			&= \bm{F}_{\bm{\omega}\bm{\omega}} - \bm{F}_{\bm{\omega}\bm{\tilde{s}}} \bm{F}_{\bm{\tilde{s}}\bm{\tilde{s}}}^{-1} \bm{F}_{\bm{\omega}\bm{\tilde{s}}}^T = \frac{2}{\sigma_n^2} \mathfrak{R} \Big\{ \sum_{t=1}^{T} \bm{\tilde{S}}^H(t)\bm{W}^H\bm{W}\bm{\tilde{S}}(t) \Big\} - \frac{2}{\sigma_n^2} \mathfrak{R} \Big\{ \sum_{t=1}^{T} \bm{\tilde{S}}^H(t)\bm{W}^H\bm{A}(\bm{A}^H\bm{A})^{-1}\bm{A}^H\bm{W}\bm{\tilde{S}}(t) \Big\} \nonumber\\
			&= \frac{2}{\sigma_n^2} \mathfrak{R} \Big\{ \sum_{t=1}^{T} \bm{\tilde{S}}^H(t)\bm{W}^H \Big( \bm{I}_M - \bm{A}(\bm{A}^H\bm{A})^{-1}\bm{A}^H \Big) \bm{W}\bm{\tilde{S}}(t) \Big\} = \frac{2}{\sigma_n^2} \mathfrak{R} \Big\{ \sum_{t=1}^{T} \bm{\tilde{S}}^H(t)\bm{W}^H \bm{P}^{\perp}_{\bm{A}} \bm{W}\bm{\tilde{S}}(t) \Big\} ,
		\end{align}
		\vspace{-0.5em}
		\hrule
		\vspace{-0.5em}
	\end{figure*}
	where $\bm{P}_{\bm{A}}^{\perp} = \bm{I}_N - \bm{A}(\bm{A}^H\bm{A})^{-1}\bm{A}^H$ denotes the orthogonal projection matrix. The second equality in the above derivation leverages the inversion property of the block matrix composed of the real and imaginary parts of a complex matrix, which is given by
	\begin{align}
		\begin{bmatrix} \mathfrak{R}\{\bm{Z}\} & -\mathfrak{I}\{\bm{Z}\} \\ \mathfrak{I}\{\bm{Z}\} & \mathfrak{R}\{\bm{Z}\} \end{bmatrix}^{-1} =  \begin{bmatrix} \mathfrak{R}\{\bm{Z}^{-1}\} & -\mathfrak{I}\{\bm{Z}^{-1}\} \\ \mathfrak{I}\{\bm{Z}^{-1}\} & \mathfrak{R}\{\bm{Z}^{-1}\} \end{bmatrix} .
	\end{align}
	Subsequently, the actual CRB matrix for the spatial angles is obtained by taking the inverse of $\bm{CRB}^{-1}(\bm{\omega})$. To comprehensively evaluate the parameter estimation performance for the $K$ targets, we consider the sum of the estimation variances for both elevation and azimuth angles, which is mathematically equivalent to taking the trace of the final CRB matrix. Therefore, the overall scalar CRB objective function, formulated with respect to the antenna positions $\bm{p}$ and the target angles $\bm{\omega}$, can be explicitly expressed as
	\begin{align} 
		\!\!\!\!{\mathrm{CRB}(\bm{p}, \bm{\omega})} \!=\!\frac{\sigma_n^2}{2}\tr \Big\{\Big[\mathfrak{R}\Big\{\sum_{t=1}^{T}{\bm{\tilde{S}}^H(t)\bm{W}^H\bm{P}_{\bm{A}}^{\perp}}\bm{W}\bm{\tilde{S}}(t)\Big\}\Big]^{-1}\!\Big\}.
	\end{align}
	\vspace{-0.5em}
	\section{Proof of Theorem \ref{theorem1}}\label{AppendixB}
	\vspace{-0.5em}
	In this section, we provide the proof of Theorem \ref{theorem1}.
	For convenience, we define a compact matrix $\bm{G}$ to simplify the structural representation of the CRB in \eqref{CRBpw} as
	\begin{align}
		\bm{G}=\frac{2}{\sigma_n^2}\sum_{t=1}^{T}\mathfrak{R}\big\{\bm{\tilde{S}}^H(t)\bm{W}^H\bm{P}_{\bm{A}}^{\perp}\bm{W}\bm{\tilde{S}}(t)\big\}\in \mathbb{C}^{2K \times 2K}.
	\end{align}
	Then we partition $\bm{G}$ into $K$ submatrices, and the $(k,k)$-th diagonal subblock $\bm{G}_{k, k} \in \mathbb{C}^{2 \times 2}$ can be expressed as
	\begin{align}\label{Gkk}
		\bm{G}_{k,k} &= \frac{2}{\sigma_n^2}\sum_{t=1}^{T}\mathfrak{R}\big\{ \big(s_k(t)\bm{I}_2\big)^H \bm{W}_k^H\bm{P}_{\bm{A}}^{\perp}\bm{W}_k \big(s_k(t)\bm{I}_2\big) \big\} \nonumber
		\\
		&=\frac{2T \sigma_k^2}{\sigma_n^2} \mathfrak{R}\big\{\bm{W}_k^H\bm{P}_{\bm{A}}^{\perp}\bm{W}_k\big\},
	\end{align}
	where $\bm{W}_k = \big[\dot{\bm{a}}_{u,k}, \dot{\bm{a}}_{v,k}\big] \in \mathbb{C}^{N \times 2}$, and
	\begin{align}
		\dot{\bm{a}}_{u,k} &= \mathrm{j}\kappa \bm{D}_x \bm{a}(\bm{p}, \bm{r}_k),\quad
		\dot{\bm{a}}_{v,k} = \mathrm{j}\kappa \bm{D}_y \bm{a}(\bm{p}, \bm{r}_k),
	\end{align}
	where $\bm{D}_x = \diag(x_1, x_2, \ldots, x_N)$ and $\bm{D}_y = \diag(y_1, y_2, \ldots, y_N)$ denote the diagonal matrices of the $x$-axis and $y$-axis MA position coordinates respectively.
	
	We define the single-target projection matrix as $\bm{P}_{\bm{a}_k}^{\perp}=\bm{I}_N-\frac{1}{N}\bm{a}(\bm{p}, \bm{r}_k)\bm{a} (\bm{p}, \bm{r}_k)^H$. Because the multi-target null space is a subset of the single-target null space, a positive semi-definite partial order relation exists between the projection operators
	\begin{align}\label{preceq}
		\bm{P}_{\bm{A}}^{\perp}\preceq\bm{P}_{\bm{a}_k}^{\perp}.
	\end{align}
	Substituting this partial order relation into \eqref{Gkk} yields
	\begin{align}
		\bm{G}_{k,k}^{-1}\succeq\bm{\bar{G}}_{k,k}^{-1}
		&=\frac{\sigma_n^2}{2T\sigma_k^2}\Big[\mathfrak{R}\big\{\bm{D}_k^H
		\bm{P}_{\bm{a}_k}^{\perp}\bm{D}_k\big\}\Big]^{-1}\\
		&=\frac{\sigma_n^2}{2T\sigma_k^2}
		\begin{bmatrix} 
			\dot{\bm{a}}_{u,k}^H
			\bm{P}_{\bm{a}_k}^{\perp}
			\dot{\bm{a}}_{u,k}
			&
			\dot{\bm{a}}_{u,k}^H
			\bm{P}_{\bm{a}_k}^{\perp}
			\dot{\bm{a}}_{v,k}
			\\[1mm]
			\dot{\bm{a}}_{v,k}^H
			\bm{P}_{\bm{a}_k}^{\perp}
			\dot{\bm{a}}_{u,k}
			&
			\dot{\bm{a}}_{v,k}^H
			\bm{P}_{\bm{a}_k}^{\perp}
			\dot{\bm{a}}_{v,k}
		\end{bmatrix}^{-1}, \nonumber
	\end{align}
	where the individual elements of the matrix can be expressed as \eqref{AA}, \eqref{BB}, \eqref{BA}, and \eqref{AB}, respectively.
	\begin{figure*}
		\vspace{-1.5em}
		\begin{align}\label{AA}
			\dot{\bm{a}}_{u,k}^H\bm{P}_{\bm{a}_k}^{\perp}\dot{\bm{a}}_{u,k}&=\dot{\bm{a}}_{u,k}^H\Big(\bm{I}_N-\frac{1}{N}\bm{a}_k\bm{a}_k^H\Big)
			\dot{\bm{a}}_{u,k}=\dot{\bm{a}}_{u,k}^H
			\dot{\bm{a}}_{u,k}-\frac{1}{N}\Big(
			\dot{\bm{a}}_{u,k}^H\bm{a}_k\Big)\Big(\bm{a}_k^H\dot{\bm{a}}_{u,k}\Big)=
			\frac{4\pi^2}{\lambda^2}
			\Big[N\overline{x^2}-N\bar{x}^2\Big]=\frac{4\pi^2 N}{\lambda^2}\var(\bm{x}),\\
			\dot{\bm{a}}_{v,k}^H\bm{P}_{\bm{a}_k}^{\perp}\dot{\bm{a}}_{v,k}&=\dot{\bm{a}}_{v,k}^H\Big(\bm{I}_N-\frac{1}{N}\bm{a}_k\bm{a}_k^H\Big)
			\dot{\bm{a}}_{v,k}= \frac{4\pi^2}{\lambda^2}\left[N\overline{y^{2}}-N\bar{y}^2\right]=\frac{4\pi^2 N}{\lambda^2}\var(\bm{y}), \label{BB}\\
			\dot{\bm{a}}_{u,k}^H\bm{P}_{\bm{a}_k}^{\perp}\dot{\bm{a}}_{v,k}&=
			\dot{\bm{a}}_{u,k}^H\dot{\bm{a}}_{v,k}-\frac{1}{N}\big(\dot{\bm{a}}_{u,k}^H\bm{a}_k\big)\big(\bm{a}_k^H\dot{\bm{a}}_{v,k}\big)=\frac{4\pi^2}{\lambda^2}\Big[ \sum\nolimits_{n=1}^{N}x_n y_n-N\bar{x}\bar{y}\Big]=\frac{4\pi^2 N}{\lambda^2}\cov(\bm{x},\bm{y}),\label{BA}\\
			\dot{\bm{a}}_{v,k}^H\bm{P}_{\bm{a}_k}^{\perp}\dot{\bm{a}}_{u,k}&=
			\dot{\bm{a}}_{v,k}^H\dot{\bm{a}}_{u,k}-\frac{1}{N}\big(\dot{\bm{a}}_{v,k}^H\bm{a}_k\big)\big(\bm{a}_k^H\dot{\bm{a}}_{u,k}\big)=\frac{4\pi^2}{\lambda^2}\Big[ \sum\nolimits_{n=1}^{N}x_n y_n-N\bar{x}\bar{y}\Big]=\frac{4\pi^2 N}{\lambda^2}\cov(\bm{x},\bm{y}),\label{AB}
		\end{align}
		\vspace{-0.5em}
		\hrule
		\vspace{-0.5em}
	\end{figure*}
	where $\bar{x} = \frac{1}{N}\sum_{n=1}^{N}x_n$ and $\overline{x^2} = \frac{1}{N}\sum_{n=1}^{N}x_n^2$. Therefore, we can obtain the LB of the CRB as
	\begin{align}
		{\mathrm{CRB}(\bm{p}, \bm{\omega})}&= \tr\{ \bm{G}^{-1}\}\ge \sum_{k=1}^{K}
		\tr\{\bm{G}_{k,k}^{-1}\}\ge  \sum_{k=1}^{K}	\tr\{\bm{\bar{G}}_{k,k}^{-1}\} \nonumber
		\\
		&=\sum_{k=1}^{K} \varpi_k
		\Big[\frac{1}{\var(\bm{x})-
			\frac{\cov^2(\bm{x},\bm{y})}
			{\var(\bm{y})}}+\frac{1}
		{\var(\bm{y})-\frac{\cov^2(\bm{x},\bm{y})}{\var(\bm{x})}}\Big]\nonumber\\
		&\ge \sum_{k=1}^{K} \frac{4\varpi_k}
		{\var(\bm{y})-\frac{\cov^2(\bm{x},\bm{y})}{\var(\bm{x})}+\var(\bm{x})-\frac{\cov^2(\bm{x},\bm{y})}{\var(\bm{y})}}\nonumber\\
		& \ge \frac{1}{\var(\bm{y}) +\var(\bm{x})} \sum_{k=1}^{K} {4\varpi_k} \triangleq \mathrm{LB}(\bm{p},\bm{\omega}),
	\end{align}
	where $\varpi_k=\frac{\sigma_n^2\lambda^2}{8\pi^2 TN \sigma_k^2}$. The first inequality utilizes the inverse property of positive definite matrices, the second inequality is based on the partial order relation in \eqref{preceq}, the third inequality applies the basic arithmetic inequality $\frac{1}{a}+\frac{1}{b} \ge \frac{4}{a+b},(a,b >0)$, and the fourth inequality holds because $\var(\bm{x}), \var(\bm{y})\ge 0$.
	\vspace{-0.5em}
	
	
	
	%
	\raggedbottom 
	\bibliographystyle{IEEEtran}
	\bibliography{mybib_Abbreviation_1}

@ARTICLE{ModelingMA,
	author={Zhu, Lipeng and Ma, Wenyan and Zhang, Rui},
	journal={IEEE Trans. Wireless Commun.}, 
	title={Modeling and Performance Analysis for Movable Antenna Enabled Wireless Communications}, 
	year={2024},
	volume={23},
	number={6},
	pages={6234-6250},
	doi={10.1109/TWC.2023.3330887}}

@misc{mao2025movableantennaarrayenhancedmultitarget,
	title={Movable-Antenna Array Enhanced Multi-Target Sensing: {CRB} Characterization and Optimization}, 
	author={Haobin Mao and Lipeng Zhu and Wenyan Ma and Zhenyu Xiao and Xiang-Gen Xia and Rui Zhang},
	year={2025},
	eprint={2511.18907},
	archivePrefix={arXiv},
	primaryClass={eess.SP},
	url={https://arxiv.org/abs/2511.18907}, 
}

@misc{chen2026exploitingmovableantennasmultitarget,
	title={Exploiting Movable Antennas for Multi-Target Wireless Sensing System in Multipath Environments}, 
	author={Xiang Chen and Ming-Min Zhao and Lebin Chen and Jie Xu and Min-Jian Zhao},
	year={2026},
	eprint={2607.23249},
	archivePrefix={arXiv},
	primaryClass={cs.IT},
	url={https://arxiv.org/abs/2607.23249}, 
}

@ARTICLE{mosek2019,
	author = {MOSEK ApS},
	title = {The {MOSEK} optimization toolbox for {MATLAB}},
	year = {2019},
	journal = {User’s Guide and
	Reference Manual},
	month={Sep.},
}

@misc{ChengzhiTWCMIMOOFDM,
      title={Tensor Decomposition-Based Wireless Sensing for {MIMO-OFDM} {ISAC} via Flexible Spatial-Temporal-Spectral Optimization}, 
author={Chengzhi Ye and others},
year={2026},
eprint={2608.15792},
archivePrefix={arXiv},
primaryClass={eess.SP},
url={https://arxiv.org/abs/2608.15792}, 
}

@ARTICLE{10944486,
	author={Jiang, Chengjun and Zhang, Chensi and Huang, Chongwen and Ge, Jianhua and Niyato, Dusit and Yuen, Chau},
	journal={IEEE Trans. Wireless Commun.}, 
	title={Movable Antenna-Assisted Integrated Sensing and Communication Systems}, 
	year={2025},
	volume={24},
	number={8},
	pages={6397-6412},
	doi={10.1109/TWC.2025.3552843},
	month={Aug.},}

@INPROCEEDINGS{YuchenMASensing,
	author={Li, Yuchen and Sun, Chen},
	booktitle={2026 IEEE Wireless Communications and Networking Conference (WCNC)}, 
	title={Movable Antenna Advanced Wireless Sensing with Scannable Angle Range}, 
	year={2026},
	volume={},
	number={},
	pages={1-6},
	doi={10.1109/WCNC65185.2026.11555046},
	month={Apr.},}

@ARTICLE{LiJinTRMAsensing,
	author={Li, Jin and Zhu, Lipeng and Liu, Nan and Zhang, Rui},
	journal={IEEE Wireless Commun. Lett.}, 
	title={Joint Transmit and Receive Antenna Position Optimization for Movable Antenna Array Enhanced Wireless Sensing}, 
	year={2025},
	volume={14},
	number={12},
	pages={3867-3871},
	doi={10.1109/LWC.2025.3603217},
	month={Dec.},}

@ARTICLE{6DSensing,
	author={Shao, Xiaodan and Zhang, Rui and Schober, Robert},
	journal={IEEE Wireless Commun. Lett.}, 
	title={Exploiting Six-Dimensional Movable Antenna for Wireless Sensing}, 
	year={2025},
	volume={14},
	number={2},
	pages={265-269},
	doi={10.1109/LWC.2024.3487966},
	month={Feb.},}

@ARTICLE{ZhuMAsurvey,
	author={Zhu, Lipeng and Ma, Wenyan and Mei, Weidong and Zeng, Yong and Wu, Qingqing and Ning, Boyu and Xiao, Zhenyu and Shao, Xiaodan and Zhang, Jun and Zhang, Rui},
	journal={IEEE Commun. Surveys Tuts.}, 
	title={A Tutorial on Movable Antennas for Wireless Networks}, 
	year={2026},
	volume={28},
	number={},
	pages={3002-3054},
	doi={10.1109/COMST.2025.3546373},
	month={Feb.},}

@ARTICLE{MaISACposition,
	author={Ma, Wenyan and Zhu, Lipeng and Zhang, Rui},
	journal={IEEE Trans. Signal Process.}, 
	title={Movable Antenna Enhanced Integrated Sensing and Communication Via Antenna Position Optimization}, 
	year={2026},
	volume={74},
	number={},
	pages={1522-1537},
	doi={10.1109/TSP.2026.3674463},
	month={Jan.},}

@ARTICLE{VaidyanathanCoPrime,
	author={Vaidyanathan, Palghat P. and Pal, Piya},
	journal={IEEE Trans. Signal Process.}, 
	title={Sparse Sensing With Co-Prime Samplers and Arrays}, 
	year={2011},
	volume={59},
	number={2},
	pages={573-586},
	doi={10.1109/TSP.2010.2089682},
	month={Feb.},}

@ARTICLE{NestedMIMO,
	author={Yang, Minglei and Sun, Lei and Yuan, Xin and Chen, Baixiao},
	journal={IEEE Signal Process. Lett.}, 
	title={A New Nested {MIMO} Array With Increased Degrees of Freedom and Hole-Free Difference Coarray}, 
	year={2018},
	volume={25},
	number={1},
	pages={40-44},
	doi={10.1109/LSP.2017.2766294},
	month={Jan.},}

@ARTICLE{DaiQLMIMOOFDM,
	author={Dai, Qianglong and others},
	journal={IEEE Commun. Surveys Tuts.}, 
	title={A Tutorial on {MIMO-OFDM} {ISAC}: From Far-Field to Near-Field}, 
	year={2026},
	volume={28},
	number={},
	pages={4319-4358},
	doi={10.1109/COMST.2025.3650568},
	month={Jan.},
}

@ARTICLE{GaoLowComplexityMIMO,
	author={Gao, Xinyu and Dai, Linglong and Sayeed, Akbar M.},
	journal={IEEE Commun. Mag.}, 
	title={Low {RF}-Complexity Technologies to Enable Millimeter-Wave {MIMO} with Large Antenna Array for 5G Wireless Communications}, 
	year={2018},
	volume={56},
	number={4},
	pages={211-217},
	doi={10.1109/MCOM.2018.1600727},
	month={Apr.},}

@ARTICLE{HaiquanXLMIMO,
	author={Lu, Haiquan and Zeng, Yong and You, Changsheng and Han, Yu and Zhang, Jiayi and Wang, Zhe and Dong, Zhenjun and Jin, Shi and Wang, Cheng-Xiang and Jiang, Tao and You, Xiaohu and Zhang, Rui},
	journal={IEEE Commun. Surveys Tuts.}, 
	title={A Tutorial on Near-Field {XL-MIMO} Communications Toward {6G}}, 
	year={2024},
	volume={26},
	number={4},
	pages={2213-2257},
	doi={10.1109/COMST.2024.3387749},
	month={Apr.},}

@ARTICLE{RuoyuDynamic,
	author={Zhang, Ruoyu and Chen, Guangyi and Cheng, Lei and Guan, Xinrong and Wu, Qingqing and Wu, Wen and Zhang, Rui},
	journal={IEEE Trans. Wireless Commun.}, 
	title={Tensor-Based Channel Estimation for Extremely Large-Scale {MIMO-OFDM} With Dynamic Metasurface Antennas}, 
	year={2025},
	volume={24},
	number={7},
	pages={6052-6068},
	doi={10.1109/TWC.2025.3551144},
	month={Jul.},}

@ARTICLE{RuoyuDOA,
	author={Zhang, Ruoyu and Shim, Byonghyo and Wu, Wen},
	journal={IEEE Trans. Signal Process.}, 
	title={Direction-of-Arrival Estimation for Large Antenna Arrays With Hybrid Analog and Digital Architectures}, 
	year={2022},
	volume={70},
	number={},
	pages={72-88},
	doi={10.1109/TSP.2021.3119768},
	month={Jan.},}

@ARTICLE{ZhangChannel2024,
	author={Zhang, Ruoyu and Cheng, Lei and Zhang, Wei and Guan, Xinrong and Cai, Yueming and Wu, Wen and Zhang, Rui},
	journal={IEEE Wireless Commun. Lett.}, 
	title={Channel Estimation for Movable-Antenna {MIMO} Systems via Tensor Decomposition}, 
	year={2024},
	volume={13},
	number={11},
	pages={3089-3093},
	doi={10.1109/LWC.2024.3450592},
	month={Nov.},}

@misc{ma2025movableantennatrajectory,
	title={Movable-Antenna Trajectory Optimization for Wireless Sensing: {CRB} Scaling Laws over Time and Space}, 
	author={Wenyan Ma and Lipeng Zhu and Rui Zhang},
	year={2025},
	eprint={2509.14905},
	archivePrefix={arXiv},
	primaryClass={cs.IT},
	url={https://arxiv.org/abs/2509.14905},
	month={Sep.}, 
}

@misc{ye2026selfcalibrationdoaestimationmovable,
 	title={Self-Calibration {DOA} Estimation for Movable Antenna Systems with Antenna Position Errors}, 
 	author={Chengzhi Ye and Ruoyu Zhang and Wen Wu and Byonghyo Shim},
 	year={2026},
 	eprint={2605.23140},
 	archivePrefix={arXiv},
 	primaryClass={eess.SP},
 	url={https://arxiv.org/abs/2605.23140},
	month={May}, 
 }

@ARTICLE{MaMAsurvey,
 	author={Ma, Wenyan and others},
 	journal={IEEE Commun. Surveys Tuts.}, 
 	title={A Survey on Reconfigurable and Movable Antennas for Wireless Communications and Sensing}, 
 	year={2026},
 	volume={28},
 	number={},
 	pages={4842-4882},
 	doi={10.1109/COMST.2026.3663021},
	month={Feb.},}

@INPROCEEDINGS{WangNearfieldMA,
	author={Wang, Yushen and Mei, Weidong and Wei, Xin and Ning, Boyu and Chen, Zhi},
	booktitle={2025 IEEE International Conference on Communications Workshops (ICC Workshops)}, 
	title={Antenna Position Optimization for Movable Antenna-Empowered Near-Field Sensing}, 
	year={2025},
	volume={},
	number={},
	pages={324-329},
	doi={10.1109/ICCWorkshops67674.2025.11162470},
	month={Jun.},}

@ARTICLE{ChengzhiIOTJ20262D,
	author={Ye, Chengzhi and Zhang, Ruoyu and Yao, Lei and Wu, Wen},
	journal={IEEE Internet Things J.}, 
	title={Joint Shape-Position Optimization Enhanced {2D} {DOA} Estimation in Movable Antenna Systems}, 
	year={2026},
	volume={13},
	number={14},
	pages={32447-32450},
	doi={10.1109/JIOT.2026.3689968},
	month={Jul.},}

@ARTICLE{PiyaVaidyanathan2010,
	author={Pal, Piya and Vaidyanathan, P. P.},
	journal={IEEE Trans. Signal Process.}, 
	title={Nested Arrays: A Novel Approach to Array Processing With Enhanced Degrees of Freedom}, 
	year={2010},
	volume={58},
	number={8},
	pages={4167-4181},
	doi={10.1109/TSP.2010.2049264},
	month={Aug.},}

@ARTICLE{C11456856,
	author={Chen, Guangyi and Zhang, Ruoyu and Guan, Xinrong and Wu, Qingqing and Ning, Boyu and Zhang, Yu and Wu, Wen and Zhang, Rui},
	journal={IEEE Trans. Wireless Commun.}, 
	title={Movable Antenna-Enabled {MIMO} Integrated Sensing and Communication: A Unified Mutual Information Framework}, 
	year={2026},
	volume={25},
	number={},
	pages={14440-14454},
	doi={10.1109/TWC.2026.3675871},
	month={Mar.},}

@ARTICLE{DOAGPEChengzhi2025,
	author={Ye, Chengzhi and Zhang, Ruoyu and Hu, Changcheng and Yao, Lei and Wu, Wen and Yuen, Chau},
	journal={IEEE Wireless Commun. Lett.}, 
	title={Robust {DOA} Estimation for Movable Antenna Arrays With Partial Gain and Phase Errors}, 
	year={2026},
	volume={15},
	number={},
	pages={545-549},
	doi={10.1109/LWC.2025.3630221},
	month={Nov.},}

@ARTICLE{6DMovableAntennaShaoXiaodan2025,
	author={Shao, Xiaodan and Jiang, Qijun and Zhang, Rui},
	journal={IEEE Trans. Wireless Commun.}, 
	title={{6D} Movable Antenna Based on User Distribution: Modeling and Optimization}, 
	year={2025},
	volume={24},
	number={1},
	pages={355-370},
	doi={10.1109/TWC.2024.3492195},
	month={Jan.},}

@ARTICLE{LWangMovableAntenna2025,
	author={Lyu, Wanting and others},
	journal={IEEE Trans. Wireless Commun.}, 
	title={Movable Antenna Enabled Integrated Sensing and Communication}, 
	year={2025},
	volume={24},
	number={4},
	pages={2862-2875},
	doi={10.1109/TWC.2025.3525631},
	month={Apr.},}

@ARTICLE{LiQiangAccurate2019,
	author={Li, Qiang and Su, Tao and Wu, Kai},
	journal={IEEE Commun. Lett.}, 
	title={Accurate {DOA} Estimation for Large-Scale Uniform Circular Array Using a Single Snapshot}, 
	year={2019},
	volume={23},
	number={2},
	pages={302-305},
	doi={10.1109/LCOMM.2018.2889855},
	month={Feb.},}

@ARTICLE{ye2026rotatableantennaenhancedwirelesssensing,
	author={Ye, Chengzhi and Zhang, Ruoyu and Du, Jincheng and Ma, Wenyan and Wu, Qingqing and Wu, Wen and Zhang, Rui},
	journal={IEEE Wireless Commun. Lett.},
	title={Rotatable Antenna-Enhanced Wireless Sensing with Uniform Sparse Array via Tensor Decomposition},
	year={2026},
	volume={},
	number={},
	pages={1-5},
	doi={10.1109/LWC.2026.3722674},
	note={early access, doi: 10.1109/LWC.2026.3722674}
}

@article{maMovableAntennaEnhanced2024,
  author={Ma, Wenyan and Zhu, Lipeng and Zhang, Rui},
journal={IEEE Trans. Wireless Commun.}, 
title={Movable Antenna Enhanced Wireless Sensing via Antenna Position Optimization}, 
year={2024},
volume={23},
number={11},
pages={16575-16589},
doi={10.1109/TWC.2024.3443293},
	month={Nov.},}

@ARTICLE{RZhangTerahertzIntegrated2025,
	author={Zhang, Ruoyu and Wu, Wen and Chen, Xiaoming and Gao, Zhen and Cai, Yueming},
	journal={IEEE Vehicular Technology Magazine}, 
	title={Terahertz Integrated Sensing and Communication-Empowered UAVs in 6G: A Transceiver Design Perspective}, 
	year={2025},
	volume={},
	number={},
	pages={2-11},
	doi={10.1109/MVT.2025.3531088},
	month={Feb.},}

@ARTICLE{Zhang2022Tensor,  
	author={Zhang, Ruoyu and Cheng, Lei and Wang, Shuai and Lou, Yi and Wu, Wen and Ng, Derrick Wing Kwan},  
	journal={IEEE Trans. Commun.}, 
	title={{Tensor decomposition-based channel estimation for hybrid mmWave massive {MIMO} in high-mobility scenarios}},   
	year={2022},  
	month = {Sep.},
	volume={70},  
	number={9},  
	pages={6325-6340},  
	doi={10.1109/TCOMM.2022.3187780}
}

@ARTICLE{Saad2020AVision,  
	author={Saad, Walid and Bennis, Mehdi and Chen, Mingzhe},  journal={IEEE Network},   
	title={A Vision of {6G} Wireless Systems: {A}pplications, Trends, Technologies, and Open Research Problems},   
	year={2020},  
	month={May}, 
	volume={34},  
	number={3},  
	pages={134-142},  
	doi={10.1109/MNET.001.1900287}
}

@article{faisal2020ultramassive,
	title={Ultramassive {MIMO} systems at {T}erahertz bands: Prospects and challenges},
	author={Faisal, Alice and Sarieddeen, Hadi and Dahrouj, Hayssam and Al-Naffouri, Tareq Y and Alouini, Mohamed-Slim},
	journal={IEEE Veh. Technol. Mag.},
	volume={15},
	number={4},
	pages={33--42},
	year={2020},
	month={Apr.},
	publisher={IEEE}
}

@ARTICLE{Zhang2024TWC,
	author={Zhang, Ruoyu and others},
	journal={IEEE Trans. Wireless Commun.}, 
	title={Integrated Sensing and Communication with Massive {MIMO}: A Unified Tensor Approach for Channel and Target Parameter Estimation}, 
	year={2024},
	month={Aug.},
	volume={23},
	number={8},
	pages={8571-8587},
	doi={10.1109/TWC.2024.3351856}
}

@ARTICLE{Zhu2023Movable,
	author={Zhu, Lipeng and Ma, Wenyan and Ning, Boyu and Zhang, Rui},
	journal={IEEE Trans. Wireless Commun. }, 
	title={Movable-Antenna Enhanced Multiuser Communication via Antenna Position Optimization}, 
	year={2024},
	month={Jul.},
	volume={23},
	number={7},
	pages={7214-7229},
}

\end{document}